\documentclass[3p,10pt,a4paper,twoside,fleqn,sort&compress]{elsarticle}
\usepackage{amssymb,amsmath,latexsym}
\usepackage{pst-all}
\usepackage[varg]{pxfonts}
\usepackage{mathrsfs}
\usepackage{shuffle}

\newtheorem{theorem}{Theorem}[section]
\newtheorem{lemma}[theorem]{Lemma}

\newtheorem{example}[theorem]{Example}
\newtheorem{proposition}[theorem]{Proposition}

\def\endproof{\qed\endtrivlist}
\expandafter\let\csname endproof*\endcsname=\endproof
\def\qedsymbol{\ifmmode\bgroup\else$\bgroup\aftergroup$\fi
  \vcenter{\hrule\hbox{\vrule height.6em\kern.6em\vrule}\hrule}\egroup}
\def\qed{\ifmmode\else\unskip\nobreak\fi\quad\qedsymbol}

\renewcommand{\cal}{\mathscr}
\renewcommand{\iff}{\Leftrightarrow}
\renewcommand{\implies}{\Rightarrow}

\renewcommand{\le}{\leqslant}
\renewcommand{\ge}{\geqslant}

\begin{document}

\journal{a journal}

\title{\Large\bf Construction of fuzzy automata from fuzzy regular expressions\tnoteref{t1}}
\tnotetext[t1]{Research supported by Ministry  of Education and Science, Republic
of Serbia, Grant No. 174013}
\author{Aleksandar Stamenkovi\' c}
\ead{aca@pmf.ni.ac.rs}

\author{Miroslav \'Ciri\'c\corref{cor}}
\ead{miroslav.ciric@pmf.edu.rs}

\cortext[cor]{Corresponding author. Tel.: +38118224492; fax: +38118533014.}

\address{University of Ni\v s, Faculty of Sciences and Mathematics, Vi\v segradska 33, 18000 Ni\v s, Serbia}

\begin{abstract}\small
Li and Pedrycz [Y. M. Li, W. Pedrycz, Fuzzy finite automata and fuzzy regular
expressions with membership values in lattice ordered monoids, Fuzzy Sets and Systems 156 (2005) 68--92] have proved fundamental results that provide different equivalent ways to represent fuzzy languages with  membership values in a lattice-ordered monoid, and~gener\-alize the well-known results of the classical theory of formal languages.~In particular, they~have~shown that a fuzzy language over an integral lattice-ordered monoid can be represented by a fuzzy regular~expres\-sion if and only if it can be recognized by a fuzzy finite automaton.~However, they did not give any effective method for constructing~an equivalent fuzzy finite automaton from a given fuzzy regular~expres\-sion.~In this paper we provide such an effective method.

Trans\-form\-ing scalars appearing in a fuzzy regular expression $\alpha $ into letters of the new extended alphabet, we convert the fuzzy regular expression $\alpha $ to an ordinary regular expression $\alpha_R$.~Then, starting from an arbitrary nondeterministic finite  automaton $\cal A$ that recognizes the language $\|\alpha_R\|$ repre\-sented by the regular expression $\alpha_R$, we construct fuzzy finite automata ${\cal A}_\alpha $ and ${\cal A}_\alpha^\mathrm{r}$ with the same or even less number~of~states than the automaton $\cal A$, which recognize the fuzzy language $\|\alpha\|$ represented by the fuzzy regular expression~$\alpha $.~The starting nondeterministic finite automaton $\cal A$ can be obtained from $\alpha_R$ using any of the well-known constructions for converting regular expressions to nondeterministic finite automata, such as Glushkov-McNaughton-Yamada's position automaton, Brzozowski's derivative automaton, Antimirov's partial derivative automaton, or Ilie-Yu's follow automaton.

\end{abstract}


\maketitle

\section{Introduction}\label{sec:in}

Study of fuzzy automata and languages was initiated in 1960s by Santos \cite{Santos.68,Santos.72,Santos.76}, Wee \cite{Wee.67}, Wee and~Fu \cite{Wee.Fu.69}, and Lee and Zadeh \cite{LZ.69}. From late 1960s until early 2000s mainly fuzzy automata and languages~with membership values in the G\"odel structure have been considered (see for example \cite{DP.80,GSG.77,MM.02}).~The idea of studying~fuzzy automata with membership values in some structured abstract set comes back to W. Wechler \cite{Wechler.78}, and in recent~years~researcher's  attention has been aimed mostly to fuzzy automata with membership values in complete residuated lattices, lattice-ordered moinoids, and other kinds of lattices.~Fuzzy automata taking membership values in a complete residuated
lattice were first studied in \cite{Qiu.01,Qiu.02}, where some basic concepts
have been discussed, and later,~extensive
research of these fuzzy automata  has been carried out~in \cite{Qiu.04,Qiu.06,WQ.10,XQ.09a,XQ.09b,XQL.09,XQLF.07}.~From a different point
of view, fuzzy automata~with membership values in a complete residuated lattice were
studied in \cite{CSIP.07,CSIP.10,IC.10,ICB.08,ICB.09,ICBP.10,SCI.11}.~Fuzzy automata with membership values in a lattice-ordered monoid~have been inves\-tigated in \cite{LL.07,LL.06,LP.05,3Li.06}, fuzzy automata over other types of lattices were the subject of
\cite{DV.10,Li.11,LP.07,KSY.07,KL.07,P.04a,P.04b,PK.04,PZ.08}, and automata which generalize fuzzy automata over any type of lattices, as well as weighted automata over semirings, have been studied recently in \cite{CDIV.10,DSV.10,JIC.11}.~It is worth noting that fuzzy automata and languages are widely used in lexical analysis, description of natural and programming languages, learning systems, control systems, neural networks, clinical monitoring, pattern recognition, databases, discrete event systems, and many other areas.

Li and Pedrycz \cite{LP.05} have proved fundamental results that provide different equivalent ways to represent fuzzy languages with  membership values in a lattice-ordered monoid, e.g., by fuzzy finite automata, crisp-deterministic fuzzy finite automata, fuzzy regular expressions, and fuzzy regular grammars.~These~results generalize the well-known results of the classical theory of formal languages.~In particular, they~have~shown that a fuzzy language over an integral lattice-ordered monoid can be represented by a fuzzy regular~expres\-sion if and only if it can be recognized by a fuzzy finite automaton.~However, Li and Pedrycz did not give any effective method for constructing an equivalent fuzzy finite automaton from a given fuzzy regular~expres\-sion.~The purpose of the present paper is to provide such an effective method.

Our basic idea is to convert a fuzzy regular expression $\alpha $ into an ordinary regular expression $\alpha_R$,~trans\-form\-ing scalars appearing in the fuzzy regular expression $\alpha $ into letters of the new extended alphabet.~Then, starting from an arbitrary nondeterministic finite  automaton $\cal A$ that recognizes the language $\|\alpha_R\|$ repre\-sented by the regular expression $\alpha_R$, we construct a fuzzy finite automaton ${\cal A}_\alpha
$ with the same number~of~states as the automaton $\cal A$, which recognizes the fuzzy language $\|\alpha\|$ represented by the fuzzy regular expression~$\alpha $. Moreover, we construct a reduced version ${\cal A}_\alpha^{\mathrm{r}}$ of the fuzzy automaton
${\cal A}_\alpha $, a fuzzy finite automaton which also recognizes the fuzzy
language $\|\alpha\|$ and can have even smaller number of states than~${\cal A}_\alpha $.~The~method~is\break generic, which means that it can be used in combination with any method for constructing a nondetermi\-nis\-tic finite automaton from an ordinary regular expression.~In the past, many different~tech\-niques for~constructing nondeterministic finite automata from regular expressions have been proposed.~Besides {\it Thompson's con\-struc\-tion\/} \cite{T.68}, which build nondeterministic finite automata with $\varepsilon$-transitions, other well-known constructions build nondeterministic finite automata without $\varepsilon$-transitions.~The~best~known~and~most used such constructions are the {\it position~automaton\/}, discovered independently by Glushkov \cite{G.61} and McNaughton and
Yamada \cite{MY.60}, Brzozowski's {\it derivative automaton\/} \cite{B.64}, Antimirov's {\it partial derivative automaton\/}~\cite{A.96}, and Ilie and Yu's {\it follow automaton\/} \cite{IY.02a,IY.02,IY.03a,IY.03}.~Each of these constructions can serve as a basis for the construction of~our fuzzy finite automata.~More information on the algorithms for building small nondeterministic finite automata from regular expressions can be found in \cite{IY.03a}.

It should be noted that
the same idea of treating scalars appearing in a fuzzy regular expression~as~the~letters of a new extended alphabet, and then treating a fuzzy regular expression
as an ordinary regular~expres\-sion
over a larger alphabet, has been recently used by Kuske \cite{K.08} in the
context of weighted regular~expres\-sions and weighted finite automata over semirings.
However, there are some significant differences between his and our approach.~First,
Kuske considered only weighted regular expressions that define proper power
series, i.e., power series with zero as the coefficient of the empty word.~In terms of the theory of fuzzy lan\-guages, these are fuzzy languages which (absolutely) do not contain the empty word.~There is one~even~more important difference.~In the mentioned paper \cite{K.08}, Kuske gave
a new proof of the famous Sch\"utzenberger's theorem \cite{S.61,E.74} which asserts that the behaviors of weighted finite automata over an arbitrary semiring~are
precisely the rational formal power series, i.e., formal power series defined
by weighted regular expressions. In his proof, Kuske first converts a weighted
regular expression $E$ to a regular expression $E'$, then he starts from
an arbitrary deterministic finite automaton that recognizes the language
defined by $E'$, and~from this automaton he constructs a weighted
finite
automaton whose behavior is the formal power series
defined by $E$.~However, the number of states of deterministic finite automata obtained from regular~expressions can be exponentially larger than the lengths of the corresponding regular expressions.~For~this~reason, regular\break expressions are more often converted to nondeterministic finite automata, and the above
mentioned con\-struc\-tions
outputs nondeterministic finite automata
whose number of states is equal to the length of the regular expression plus one, or even less than that number.~In addition, our constructions output
fuzzy~finite auto\-mata with the same or even smaller number of states than
the original nondeterministic finite automaton.

As we have said, the size of
an automaton obtained from a regular expression plays a very important~role,
and for that reason regular expressions are mostly converted to nondeterministic finite automata.~On the other hand, for practical applications deterministic finite automata are usually needed, but determinization
of a nondeterministic finite automaton can cause an exponential blow up in
the number of states.~That~is~why the number of states of a nondeterministic
finite automaton has to be reduced prior to determinization.~As the minimization of nondeterministic finite automata is computationally hard, we must be satisfied with the methods for reducing the number of states that do not necessarily give a minimal automaton, but~rather~provide a reasonably small automaton that can be effectively computed.~Such reduction methods have been recently investigated
in \cite{CSY.05,CC.04,IY.02,IY.03,IY.03a,INY.04,ISY.05}, in the context of
nondeterministic finite automata, and in \cite{CSIP.07,CSIP.10,SCI.11},~in
the context of fuzzy finite automata (see also \cite{CIBJ.11,CIDB.11,CIJD.11}).~Key role in the state reduction of nondeterministic\break finite automata play right and left invariant equivalences, which have been generalized in the fuzzy framework as right and left invariant fuzzy equivalences
(cf.~\cite{CSIP.07,CSIP.10,SCI.11}).~It is worth noting that right and left invariant (fuzzy) equivalences
are also known
as forward and backward bisimulation (fuzzy) equivalences (cf.~\cite{CIBJ.11,CIDB.11,CIJD.11}).~In
particular, it  has been proved in \cite{CZ.01a,CZ.01b,CZ.02,IY.02a,IY.02,IY.03a,IY.03} that both the partial derivative automaton and the follow
automaton are factor automata of the position
automaton with respect to certain right~invar\-iant equivalences.~State reduction of fuzzy finite automata by means of right~invariant~fuzzy~and~crisp~equi\-valences will be also considered in this paper.~Let us also note that the above mentioned
determinization problem has been recently investigated in the fuzzy framework
in \cite{B.02,CDIV.10,ICB.08,JIC.11,LP.05}.

Our main results are the following.~We start from a given fuzzy regular expression $\alpha $ over
an alphabet $X$ and a lattice-orde\-red monoid ${\cal L}=(L,\land ,\lor, \otimes,0,1,e)$, and we define an ordinary regular expression $\alpha_R$ over a new alphabet $X\cup Y$, where $Y$ consists of the letters associated with different scalars appearing in $\alpha $.~The mapping $\varphi_\alpha $ of $X\cup Y$ to $L$ , which maps all letters from $X$ to $e$, and letters from $Y$ to related scalars appearing in $\alpha $, can be extended in a natural way to a homomorphism $\varphi_\alpha^*$ of the free monoid $(X\cup Y)^*$ to the monoid $(L,\otimes ,e)$.~In the case when $\cal L$ is an integral lattice-ordered monoid, using this homomorphism we establish a relationship between the fuzzy language $\|\alpha\|$ represented by $\alpha $ and the language $\|\alpha_R\|$ represented by $\alpha_R$ (cf.~Theorem \ref{lema3}), and starting from any nondeterministic finite automaton ${\cal A}$ that recognizes the language $\|\alpha_R\|$ we define the fuzzy automaton ${\cal A}_\alpha $ associated with $\cal A$ and $\alpha $, and we prove that ${\cal A}_\alpha $ recognizes the fuzzy language $\|\alpha \|$ represented by the fuzzy regular expression $\alpha $ (cf.~Theorem \ref{th:1}).

However,~the~aforementioned definition of the fuzzy automaton ${\cal A}_\alpha $ is not  sufficiently constructive, because the computing of the fuzzy transition relation and the fuzzy set of terminal states of ${\cal A}_\alpha $ requires the computing of minimal words in certain infinite languages with respect to the embedding order, which might be a problem.~We~solve~this problem introducing a reflexive and transitive fuzzy relation
$R_{\cal A}$~on~the set of states of the starting~nondeter\-ministic finite automaton $\cal A$, which can be effectively computed~as~the $n$-th power of an easily computable fuzzy relation, where $n$ is the number~of~states~of~$\cal A$.~We~express~the fuzzy relation $R_{\cal A}$ in terms of the homomor\-phism $\varphi_\alpha^*$ and the transition relation of $\cal A$ (cf.~Theorem~\ref{lemanova1}), and then we express the fuzzy transition relation and the fuzzy set of terminal states of ${\cal A}_\alpha $ in terms of the fuzzy relation $R_{\cal A}$,  the transition relation of $\cal A$, and the set of terminal states of $\cal A$ (cf.~Theorem \ref{th:2}).~This result provides an effective construction of the fuzzy finite automaton ${\cal A}_\alpha $ associated with $\cal A$ and the fuzzy regular expression~$\alpha $.

Using the fuzzy relation~$R_{\cal A}$ we also construct a version ${\cal A}_\alpha^\mathrm{r}$ of the fuzzy finite automaton ${\cal A}_\alpha $ which can have even smaller number of states than the fuzzy automaton ${\cal A}_\alpha $ and the automaton ${\cal A}$, and recognizes the same fuzzy language $\|\alpha\|$ (cf.~Theorem \ref{th:3}).~We show by an example that the number of states of ${\cal A}_\alpha^\mathrm{r}$ can be strictly smaller than the number of states of $\cal A$ and ${\cal A}_\alpha $.~We also discuss the state reduction of the fuzzy automaton ${\cal A}_\alpha $ by means of right invariant crisp equivalences, and we show that even if the starting automaton $\cal A$ is a minimal deterministic automaton, the number of states of the fuzzy automaton ${\cal A}_\alpha $ could be reduced.~Finally, we describe certain properties of fuzzy automata obtained from the position and the follow automaton.

The structure of the paper is as follows.~In Section \ref{sec:pre} we recall some basic definitions and results~concern\-ing fuzzy sets and relations over lattice ordered monoids, nondeterministic and fuzzy automata, and regular and fuzzy regular expressions.~In Section \ref{sec:fuzzy.autom.exp} we give the basic construction of a fuzzy finite automaton ${\cal A}_\alpha $ associated with a fuzzy regular expression $\alpha $ and a nondeterministic finite automaton $\cal A$ recognizing the language $\|\alpha_R\|$.~Section \ref{sec:effect.con} addresses the issue of the effective construction of the fuzzy automaton ${\cal A}_\alpha $, and in Section \ref{sec:reduced} we deal with the version of this construction that gives a fuzzy automaton with a reduced number of states
with respect to the original construction.~Finally, in Section \ref{sec:redstate} we discuss the problem of the reduction of the number of states of fuzzy finite automata constructed from fuzzy regular expressions.

\section{Preliminaries}\label{sec:pre}

In this section we recall some basic definitions and results
concerning fuzzy sets and relations over~lattice ordered monoids, nondeterministic and fuzzy automata, and regular and fuzzy regular expressions.

\subsection{Lattice-ordered monoids}\label{sec:pre:lomon}

A {\it lattice-ordered monoid\/} or an {\it $\ell-$monoid}
\cite{L.04,LL.06,LP.05,SL.06} is an algebra ${\cal L}=(L, \land,
\lor, \otimes, 0, 1, e)$ such that

\begin{itemize}
\parskip=0pt
\item[\rm{(L1)}] $(L, \land, \lor, 0, 1)$ is a lattice with the
least element $0$ and the greatest element $1$,
\item[\rm{(L2)}] $(L, \otimes, e)$ is a monoid with the unit $e$,
\item[\rm{(L3)}] $x\otimes 0=0\otimes x=0,$ for every $x\in L,$
\item[\rm{(L4)}] $x\otimes(y\vee z)=x\otimes y\vee x\otimes z$, \ $(x\vee
y)\otimes z=x\otimes z\vee y\otimes z,$ \ for all $x,y,z\in L.$
\end{itemize}
The operation $\otimes$ is called the {\it multiplication\/}.~In addition, if $(L, \land, \lor, 0, 1)$ is a complete lattice and
satisfies the following infinite distributive laws
\begin{equation}\label{eq:inf.distr}
x\otimes\bigl(\bigvee_{i\in I}x_i\bigr)=\bigvee_{i\in
I}(x\otimes x_i),\qquad \bigl(\bigvee_{i\in I}x_i\bigr)\otimes
x=\bigvee_{i\in I}(x_i\otimes x),
\end{equation}
then $\cal L$ is called a {\it quantale\/}.~In the general case, in an $\ell $-monoid ${\cal L}=(L, \land,
\lor, \otimes, 0, 1, e)$ the greatest element~$1$ of the lattice $(L, \land, \lor, 0, 1)$ and the unit element $e$ of the monoid $(L, \otimes, e)$ are different. If $1$ and $e$ coincide, then $\cal L$ is called an {\it integral $\ell-$monoid\/}.

It can be easily verified that with respect to $\le$, the multiplication $\otimes$ in an $\ell $-monoid is isotone in both arguments, i.e., for
all $x,y,z\in L$ we have
\begin{equation}\label{isotonost}
x\le y\ \text{implies}\ x\otimes z\le y\otimes z\ \text{and}\ z\otimes x\le z\otimes y.
\end{equation}

An integral quantale with commutative multiplication is known as a {\it complete residuated lattice\/} (cf.~\cite{BEL.02,BV.05}). The most studied and applied kinds of complete residuated lattices, with the support $[0,1]$, $x\land y=\min(x,y)$ and $x\lor y=\max(x,y)$, are the {\it Lukasiewicz structure\/}, with the multiplication defined by $x\otimes y=max(x+y-1,0)$, the {\it Goguen\/} or {\it product structure\/}, with $x\otimes
y=x\cdot y$, and the {\it G\" odel structure\/}, with $x\otimes y=min(x,y)$.~The fourth important type of complete residuated lattices is the two-element Boolean algebra of classical logic with the support $\{0,1\}$, called the {\it Boolean structure\/}.

In the further text, if not noted otherwise, $\cal L$ will be an $\ell-$monoid.~A {\it
fuzzy subset\/} of a set $A$ is defined~as~any mapping from $A$ into
$L$.~Ordinary crisp subsets~of~$A$ are considered as fuzzy subsets of $A$ taking membership values in the set
$\{0,e\}\subseteq L$.~Let $f$ and $g$ be two
fuzzy subsets of $A$.~The {\it equality\/} of $f$ and $g$ is defined as the
usual equality of functions, i.e., $f=g$ if and only if $f(x)=g(x)$, for every
$x\in A$. The {\it inclusion\/} $f\leqslant g$ is also defined pointwise:~$f\leqslant g$ if
and only if $f(x)\leqslant g(x)$, for every $x\in A$.~Endowed with this partial order the set
$L^A$ of all fuzzy subsets of $A$
forms the distributive lattice, in which the meet (intersection)
$f\land g$ and the join (union) $f\lor g$ of any fuzzy subsets
$f,g$ of $A$ are also fuzzy subsets of $A$ over $\cal L$ defined
by
\begin{equation}\label{eq:and}
(f\land g)(x) =f(x)\land g(x),\mbox{\ \ }(f\lor g)(x) =f(x)\lor
g(x).
\end{equation}
for each $x\in L$.~The {\it crisp~part\/}~of a fuzzy subset $f\in L^A$ is a crisp
subset $\widehat f=\{a\in A\,|\, f(a)=e\}$ of $A$.~We~will~also consider $\widehat f$ as a mapping
$\widehat f:A\to L$ defined by $\widehat f(a)=e$, if $f(a)=e$, and
$\widehat f(a)=0$, otherwise.

A {\it fuzzy relation\/} on $A$ is any fuzzy subset of $A\times
A.$ The equality, inclusion and ordering of fuzzy relations are
defined as for fuzzy sets.~For fuzzy relations $R$ and $S$ on a set $A$,  their
{\it composition\/} $R\circ S$ is a fuzzy relation on $A$ defined~by
\begin{equation}\label{eq:komp}
(R\circ S)(a,b)=\bigvee_{c\in A}R(a,c)\otimes S(c,b),
\end{equation} for all $a,b\in A,$ and for a fuzzy subset
$f$ of $A$ and a fuzzy relation $R$ on $A$, the {\it
compositions\/} $f\circ R$ and $R\circ f$ are fuzzy subsets of $A$
defined, for any $a\in A$, by
\begin{equation}\label{eq:komp1}
(f\circ R)(a)=\bigvee_{b\in A}f(b)\otimes R(b,a),\mbox{\ \
}(R\circ f)(a)=\bigvee_{b\in A}R(a,b)\otimes f(b).
\end{equation}
For fuzzy subsets $f$ and $g$ we write
\begin{equation}\label{eq:komp2}
f\circ g=\bigvee_{a\in A}f(a)\otimes g(a).
\end{equation} It is well
known that the composition of fuzzy relations is associative. Moreover
\begin{equation}\label{eq:asoc}
(f\circ R)\circ S=f\circ (R\circ S),\quad{\ \ }(R\circ S)\circ f=R\circ (S\circ f),\quad{\ \ }(f\circ R)\circ
g=f\circ (R\circ g),
\end{equation} for all fuzzy subsets $f$ and $g$ of $A,$ and
fuzzy relations $R$ and $S$ on $A.$ If $A$ is a finite set with $n$ elements,~then $R$ and $S$ can be treated
as $n\times n$ matrices over ${\cal L}$, and $R\circ S$ is their
matrix product, whereas $f\circ R$ can be treated as the product
of the $1\times n$ matrix $f$ and the $n\times n$ matrix $R,$ and $R\circ
f$ as the product of the $n\times n$~matrix $R$ and the $n\times 1$ matrix
$f^t$ (the transpose of $f$).

For a finite set $A$ and an fuzzy relation $R$ on $A$, a fuzzy
relation $R^n$ is defined inductively as follows:~$R^0$~is the crisp
equality on $A$, and $R^{n+1}=R^n\circ R,$ for $n\in \Bbb
N\cup\{0\}.$

A fuzzy relation $R$ on $A$ is said to be
\begin{itemize}\parskip=-3pt
\parskip=-2pt
\item[(R)] {\it reflexive\/}  if $R
(a,a)=e$, for every $a\in A$; \item[(S)] {\it symmetric\/}  if $R (a,b)=R (b,a)$, for all $a,b\in A$;
\item[(T)] {\it transitive\/}  if for
all $a,b,c\in A$ we have $R (a,b)\otimes R (b,c)\le R (a,c)$.
\end{itemize}
It is easy to check that a reflexive fuzzy relation $R$ is transitive if and only if $R^2=R$, and then $R^n=R$, for every $n\in \mathbb N$.~A reflexive, symmetric and transitive fuzzy relation is called a {\it fuzzy equivalence\/}.~For a fuzzy equivalence $E  $ on $A$ and $a\in A$ we define a
fuzzy subset $E _a$ of $A$ by $E_a(x)=E(a,x)$,~for~every $x\in
A$.~We call $E _a$ the {\it equivalence class\/} of $E$
deter\-mined by $a$.~The set $A/E =\{E _a\,|\, a\in A\}$ is
called~the {\it factor set\/} of $A$ with respect to $E$ (cf.
\cite{BEL.02,BV.05,CIB.07}).~We use the same notation for crisp
equivalences, i.e., for an equivalence $\pi $ on $A$, the related
factor set is denoted by $A/\pi $, the equivalence class of an
element $a\in A$ is denoted by $\pi_a$.~A fuzzy equivalence $E$ on
a set $A$ is called a {\it fuzzy equality\/} if for all $x,y\in
A$, $E(x,y)=e$ implies $x=y$. In other words, $E$ is a fuzzy
equality if and only if its crisp part $\widehat E$ is a crisp
equality.

\subsection{Fuzzy regular expressions}\label{sec:langre}

Let $X$ be a non-empty set, which is called an {\it alphabet\/} and whose elements are called {\it letters\/}, and let $X^*$ be the free monoid over $X$, i.e., the set of all finite sequences of letters from $X$, including the empty sequence, equipped with the concatenation operation.~Elements of $X^*$ are called {\it words\/}, and the empty sequence is denoted by $\varepsilon $ and called the {\it empty word\/}.

A {\it  fuzzy language\/} in $X^*$ is defined as any fuzzy subset of $X^*$.~A {\it language\/} in $X^*$ is a fuzzy language in
$X^*$ taking membership values in the set $\{0,e\}$. For a
 fuzzy language $f$ and a scalar $\lambda\in L$, the {\it scalar
multiplication\/} $\lambda\otimes f$ is a
fuzzy language in $X^*$ defined by
\[
(\lambda\otimes f)(u)=\lambda\otimes f(u),
\]for any $u\in X^*.$
The {\it union} ({\it join}) $f\lor g$ of  fuzzy languages $f$ and
$g$ is defined as the union of  fuzzy subsets $f$ and $g$. The
{\it concatenation} ({\it product\/}) $fg$ of  fuzzy
languages $f$ and $g$ is defined by
\[
(f g)(u)=\bigvee_{u=vw}f(v)\otimes g(w).
\]
The concatenation of fuzzy languages is an associative operation, and for
$n\in \mathbb{N}$, the {\it $n$-th power\/} of a fuzzy language $f$ is defined
inductively by $f^0=f_\varepsilon $, where $f_\varepsilon$ is a characteristic function of the empty word
$\varepsilon$, i.e.,
\begin{equation}
\label{caracteristic}
f_\varepsilon(u)=\begin{cases}
\ e & \text{if}\ u=\varepsilon\\
\ 0 & \text{otherwise}
\end{cases},
\end{equation}
and $f^{n+1}=f^nf$, for each $n\in \mathbb{N}\cup \{0\}$.~The {\it Kleene closure\/} of a fuzzy language $f$, denoted by $f^*,$ is defined by
\[
f= \bigvee_{n\in \mathbb{N}\cup\{0\}}f^n.
\]

Recall the following result proved in \cite{LP.05}.

\begin{proposition}\label{prop3} If
$\cal L$ is an integral $\ell-$monoid, then for any  fuzzy
language $f$, the Kleene closure is well defined.
\end{proposition}

The family $\cal {LR}$ of {\it fuzzy regular expressions\/}  over a finite alphabet $X$
is defined inductively in the following way (cf.~\cite{LL.06,LP.05}):
\begin{itemize}
\parskip=-2pt
\item[{\rm (i)}] $\emptyset\in {\cal{LR}}$;
\item[{\rm (ii)}] $\varepsilon\in {\cal{LR}}$;
\item[{\rm (iii)}] $x\in {\cal{LR}}$, for all $x\in X$;
\item[{\rm (iv)}] $(\lambda\alpha)\in{\cal{LR}}$, for all $\lambda\in L$ and
$\alpha\in{\cal{LR}}$ ({\it scalar multiplication\/});
\item[{\rm (v)}] $(\alpha_1+\alpha_2)\in{\cal{LR}},$ for all $\alpha_1, \alpha_2
\in{\cal{LR}}$ ({\it addition\/});
\item[{\rm (vi)}] $(\alpha_1\alpha_2)\in{\cal{LR}},$ for all $\alpha_1, \alpha_2
\in{\cal{LR}}$ ({\it concatenation\/});
\item[{\rm (vii)}] $(\alpha^*)\in{\cal{LR}}$, for all $\alpha \in{\cal{LR}}$  ({\it star operation\/});
\item[{\rm (viii)}] There are no other fuzzy regular expressions than those given in steps (i)--(viii).
\end{itemize}
In order to avoid parentheses it is assumed that the star operation has the
highest priority, then concatenation and then addition.~For any fuzzy regular expression
$\alpha\in \cal{LR}$, the {\it fuzzy language $\|\alpha\|$ determined by\/} $\alpha $ is defined
inductively as follows (cf.~\cite{LL.06,LP.05}):
\begin{itemize}
\parskip=-2pt
\item[{\rm (i)}] $\|\emptyset\|(u)=0$, for every $u\in X^*$,
\item[{\rm (ii)}] For $\alpha\in X\cup\{\varepsilon\}$, $\|\alpha\|=f_{\alpha}$, where $f_{\alpha}$ is the characteristic
function of $\alpha$ defined by
\[
f_{\alpha}(u)=
\begin{cases}
\ e & \text{if}\ u=\alpha\\
\ 0 & \text{otherwise}
\end{cases};
\]
\item[{\rm (iii)}] $\|\lambda\alpha\|=\lambda\otimes\|\alpha\|$
for all $\lambda\in L$ and $\alpha\in{\cal{LR}};$ \item[{\rm
(iv)}] $\|(\alpha_1+\alpha_2)\|=\|\alpha_1\|\lor\|\alpha_2\|,$ for
all $\alpha_1, \alpha_2 \in{\cal{LR}};$ \item[{\rm (v)}]
$\|(\alpha_1\alpha_2)\|=\|\alpha_1\|\,\|\alpha_2\|,$ for all
$\alpha_1, \alpha_2 \in{\cal{LR}};$ \item[{\rm (v)}]
$\|\alpha^*\|=\|\alpha\|^*,$ for all $\alpha \in{\cal{LR}}.$
\end{itemize}
For a fuzzy regular expression $\alpha$ over $X$, the {\it
length\/} of $\alpha$, denoted by $|\alpha|_X$, is the number of
occurrences of letters from $X$ in $\alpha.$

A fuzzy regular expression $\alpha$  which does not contain any occurrence of an element of~$L$ is called a {\it regular expression\/} over an alphabet $X$.~In other words, regular expresions are those fuzzy regular expressions that are obtained without using any scalar mutiplication.~Note that the fuzzy language $\|\alpha\|$ defined by a regular
expression~$\alpha$~takes~mem\-bership values in the set $\{0,e\}$,
and thus, it can be considered as an ordinary subset of $X^*.$

For the free monoid $X^*$ we set $X^+=X^*\setminus\{ \varepsilon\}$.~The {\it length\/} of a word $u\in X^*$, in notation $|u|$, is the number of appearances of letters from $X$ in $u$.~The embedding order relation $\le_{em}$ is defined on $X^*$ by
\begin{equation}\label{eq:embedding}
u\le_{em} v\ \iff\ u=u_1u_2\cdots u_n\mbox{\ \ and\ \
}v=v_0u_1v_1u_2\cdots v_{n-1}u_nv_n,
\end{equation} where $n\in \Bbb N$
and $u,v,u_1,u_2,\dots,u_n,v_0,v_1,\dots,v_n\in X^*.$

\begin{proposition}\label{prop1}{\rm (\cite{H.69,H.52})} For any alphabet
$X$, $\le_{em}$ is a partial order on $X^*.$ Any set of pairwise
incomparable words in the partially ordered set $(X^*,\le_{em})$
is finite.

Consequently, for any $U\subseteq X^*$, the set $M(U)$ of all minimal words from $U$ with respect to $\le_{em}$ is finite.
\end{proposition}

Throughout the paper, the set of all minimal words from $U\subseteq X^*$ with respect to the embedding order~$\le_{em}$ will be denoted by $M(U)$, as in the previous proposition.

\subsection{Fuzzy automata}\label{sec:pre:fuzzy.autom}

Let $\cal L$ be an $\ell-$monoid.~A {\it  fuzzy automaton\/} (over $\cal L$) is defined as a five-tuple ${\cal
A}=(A,X,\delta^A ,\sigma^A ,\tau^A)$, where $A$ and $X$ are non-empty sets,
called respectively the {\it set of states\/} and the {\it input
alphabet\/},  $\delta^A: A\times X\times A\to L$ is a fuzzy subset
of $A\times X\times A$, called the {\it  fuzzy transition
relation\/}, $\sigma^A \in L^A$ is the  fuzzy set of {\it
initial~states\/}, and $\tau^A\in L^A$ is the  fuzzy set of~{\it
terminal states\/}.~We will assume that the input alphabet $X$ is always
finite.~A~fuzzy~auto\-maton whose set of states is finite is called a {\it  fuzzy
finite automaton\/}.~Since all fuzzy automata considered in this paper will
be finite, we will speak simply {\it fuzzy automaton\/} instead of fuzzy finite automaton.~Cardinality of a fuzzy automaton ${\cal
A}$, in notation $|{\cal A}|$, is defined as the
cardinality $|A|$ of its set of states $A$.

The fuzzy transition relation $\delta^A$ can be extended up to a mapping $\delta^A_*:
A\times X^*\times A\to L$ in the following way: If $a, b\in A$ and
$\varepsilon\in X^*$ is the empty word, then
\begin{equation}
\label{eeextension} \delta^A_* (a, \varepsilon, b)=
\begin{cases}
\ e & \text{if}\ a=b\\
\ 0 & \text{otherwise}
\end{cases},
\end{equation}
and if $a, b\in A,$ $u\in X^*$ and $x\in X,$ then
\begin{equation}
\label{eextension} \delta^A_* (a, ux, b)=\bigvee_{c\in
A}\delta^A_*(a, u, c)\otimes\delta^A(c, x, b)
\end{equation}Without danger of confusion we shall write just
$\delta^A$ instead of $\delta^A_*.$

By (L4) and Theorem 3.1 in \cite{LP.05} we have that
\begin{equation}
\label{extension} \delta^A(a, uv, b)=\bigvee_{c\in A}\delta^A(a,
u, c)\otimes\delta^A(c, v, b),
\end{equation}for all $a, b\in A$ and $u, v\in X^*.$

For any $u\in X^*$ we define a fuzzy relation $\delta^A_u\in
L^{A\times A},$ called the {\it  fuzzy transition relation\/}
determined by $u$, by $\delta^A_u(a,b)=\delta^A(a,u,b),$ for all
$a,b\in A$. Then for all $u,v\in X^*,$ the equality
(\ref{extension}) can be written as
$\delta^A_{uv}=\delta^A_u\circ\delta^A_v.$

~A {\it fuzzy language recognized by a fuzzy
automaton\/} ${\cal A}=(A,X,\delta^{A} ,\sigma^{A} ,\tau^{A} )$,
denoted by $L({\cal A})$, is a fuzzy language in $X^*$ defined by
\begin{equation}
\label{language} L({\cal A})(u)=\bigvee_{a, b\in
A}\sigma^A(a)\otimes\delta^A(a, u, b)\otimes\tau^A(b),
\end{equation}or equivalently,
\begin{equation}
\label{rellanguage} L({\cal
A})(u)=\sigma^A\circ\delta^A_u\circ\tau^A=\sigma^{A} \circ
\delta_{x_1}^{A}\circ \delta_{x_2}^{A}\circ \dots \circ
\delta_{x_n}^{A}\circ \tau^{A},
\end{equation}for any $u=x_1x_2\dots x_n\in X^*$ with $x_1,x_2,\ldots ,x_n\in
X$.

In particular, if ${\cal A}=(A, X,
\delta^A, a_0, \tau^A)$ is a fuzzy automaton having a single crisp initial state $a_0$, then the
 fuzzy language $L({\cal A})$ recognized by $\cal A$ is given by
\begin{equation} \label{languageinit}
L({\cal A})(u)=\bigvee_{a\in A}\delta^A(a_0, u,
a)\otimes\tau^A(a).
\end{equation}or equivalently,
\begin{equation}\label{rellanguageinit}
L({\cal A})(u)=(\delta^A_u\circ\tau^A)(a_0)=(\delta_{x_1}^{A}\circ
\delta_{x_2}^{A}\circ \dots \circ \delta_{x_n}^{A}\circ
\tau^{A})(a_0),
\end{equation}for any $u=x_1x_2\dots x_n\in X^*$ with $x_1,x_2,\ldots ,x_n\in
X$.

In the further text, ordinary nondeterministic automata will be considered as  fuzzy automata.~Namely,  by a nondeterministic automaton we mean a fuzzy
automaton ${\cal A}=(A, X,
\delta^A, \sigma^A, \tau^A)$ such that
$\delta^A_x$ is a~fuzzy relation taking values in the set $\{0,e\}$,   for each $x\in X$,
and $\sigma^A$ and $\tau^A$ are fuzzy sets also taking values in~$\{0,e\}$. In this case, the fuzzy language recognized by $\cal A$ is a crisp language,
and it is exactly the~language~recognized by a nondeterministic automaton in the sense of the well-known definition from the~classical~theory~of~non\-deterministic automata.

Let ${\cal A}=(A, X,\delta^A,\sigma^A,\tau^A)$  be a fuzzy automaton and let $E$ be
a fuzzy equivalence~on~$A$.~Without~any~restriction on the
 fuzzy equivalence $E$, we define a fuzzy
transition relation $\delta^{A/E}:A/E\times X\times A/E\to L$~by
\begin{equation}\label{eq:dE1}
\delta^{A/E}(E_a,x,E_b) = \bigvee_{a',b'\in A} E(a,a')\otimes
\delta (a',x,b')\otimes E(b',b)= (E\circ \delta_x\circ E)(a,b) = E_a\circ
\delta_x\circ E_b,
\end{equation}
and fuzzy sets $\sigma^{A/E}\in L^{A/E}$ and $\tau^E\in L^{A/E}$ of initial and terminal states by
\begin{gather}
\sigma^{A/E}(E_a) = \bigvee_{a'\in A}\sigma^A (a')\otimes E(a',a)
= (\sigma^A\circ
E)(a) = \sigma^A\circ E_a , \label{eq:sE} \\
\tau^{A/E}(E_a) = \bigvee_{a'\in A}\tau^A (a')\otimes E(a',a) =
(\tau^A\circ E)(a) = \tau^A\circ E_a , \label{eq:tE}
\end{gather}
for any $a\in A$. Evidently, $\delta^{A/E}$, $\sigma^{A/E}$ and $\tau^{A/E}$
are
well-defined, and${\cal
A}/E=(A/E,X,\delta^{A/E},\sigma^{A/E},\tau^{A/E})$ is a fuzzy
automaton, called the {\it factor fuzzy automaton\/} of $\cal A$
with respect to $E$.

\subsection{Position automata}\label{sec:pre:pos}

In this section we recall the construction of the position
automaton from a regular expression \cite{G.61,{MY.60}}.

Let $\alpha$ be a regular expression over an alphabet $X$. Denote
by $\overline{\alpha}$ the expression obtained from $\alpha$ by~marking each letter in $\alpha$ with its position.~The same
notation will be used for removing indices, that is, for a regular
expression $\alpha$ we put
$\alpha=\overline{\overline{\alpha}}$.~We define the following sets:
\begin{itemize}
\parskip=-2pt
\item[\rm (i)] $pos_0(\alpha)=\{0,1,\dots ,|\alpha|_X\}$,
\item[\rm (ii)] $first(\alpha)=\{i\ |\ x_iu\in
\|\overline{\alpha}\|\}$, \item[\rm (iii)] $last(\alpha)=\{i\ |\
ux_i\in \|\overline{\alpha}\|\}$, \item[\rm (iv)] $follow(\alpha,
i)=\{j\ |\ ux_ix_jv\in \|\overline{\alpha}\|\}$, \item[\rm (v)]
$follow(\alpha, 0)=first(\alpha)$, \item[\rm (vi)] $
last_0(\alpha)=\begin{cases}
                              last(\alpha), & \varepsilon\not\in \|\alpha\|\\
                              last(\alpha)\cup\{0\}, & \varepsilon\in \|\alpha\|
                              \end{cases}.
$
\end{itemize}
Define $\delta_{pos}\subseteq pos_0(\alpha)\times X\times pos_0(\alpha)$ by
\[
(i,x,j)\in \delta_{pos}\ \ \iff\ \ \overline{x_j}=x\ \ \text{and}\ \ j\in follow(\alpha, i).
\]
Then ${\cal A}_{pos}(\alpha)=(pos_0(\alpha), X, \delta_{pos}, 0,
last_0(\alpha))$ is a nondeterministic automaton called the {\it position automaton\/}~of~$\alpha$.
It was shown by Glushkov \cite{G.61} and McNaughton and Yamada
\cite{MY.60} that $L({\cal A}_{pos}(\alpha))=\|\alpha\|$.

For the sake
of simplicity, instead of  ${\cal A}_{pos}(\alpha)=(pos_0(\alpha), X, \delta_{pos}, 0,last_0(\alpha))$, in the further text we will write ${\cal A}_{\rm
p}(\alpha)=(A_{\rm p}, X, \delta^{A_{\rm p}}, 0, \tau^{A_{\rm
p}})$.

\section{Fuzzy automata from fuzzy regular
expressions:\ Basic construction}\label{sec:fuzzy.autom.exp}

For an $\ell-$monoid ${\cal L}=(L, \land, \lor, \otimes, 0, 1, e)$,  $A,B\subseteq L$ and $\lambda\in L$ we will use the following notation
\[
A\otimes B=\{a\otimes b\mid a\in A,\,b\in B\}, \qquad A\lor B=\{a\lor b\mid a\in A,\,b\in B\}, \qquad \lambda\otimes A=\{\lambda\otimes a\mid a\in A\}.
\]

The following lemma will be useful in our further work.

\begin{lemma}\label{lema1} Let ${\cal L}=(L, \land,
\lor, \otimes, 0, 1, e)$ be an $\ell-$monoid, let $A,B\subseteq L$ and $\lambda\in
L$. If there exist finite sets $C\subseteq A$ and $D\subseteq B$ such that
\[
(\forall a\in A)(\exists c\in C)\ a\le c\mbox{\ \ and\ \
}(\forall b\in B)(\exists d\in D)\ b\le d,
\]
then there exist $\bigvee A$, $\bigvee B$, $\bigvee A\otimes B$, $\bigvee A\vee B$ and $\bigvee \lambda\otimes A$, and we have that $\bigvee A=\bigvee C$, $\bigvee B=\bigvee D$, and
\[
\bigvee A\otimes B=\bigl(\bigvee A\bigr)\otimes\bigl(\bigvee B\bigr), \qquad
\bigvee A\vee B=\bigl(\bigvee A\bigr)\vee\bigl(\bigvee
B\bigr), \qquad \bigvee \lambda\otimes A=\lambda\otimes\bigl(\bigvee A\bigr).
\]
\end{lemma}

\begin{proof}
The proof of this lemma is elementary and will be omitted.
\end{proof}

Let $\cal L$ be an $\ell-$monoid and let $\alpha$ be a fuzzy
regular expression over a finite alphabet $X.$ Let $K$ be the set
of all $\lambda\in L$ appearing in $\alpha$ (if
$\alpha$ is a fuzzy regular expression without scalar multiplication then $K=\emptyset$) and
let $Y$ be an alphabet such that $Y\cap X=\emptyset$ and
$|K|=|Y|$, and let $\lambda\mapsto\lambda '$ be an arbitrary bijective mapping
from $K$ to $Y$. We will call $Y$   the {\it alphabet associated with\/}
$\alpha$.~It is clear that $Y$ is finite.

Let us denote by $\alpha_R$ the expression obtained from $\alpha$ by replacing each
$\lambda\in K$ by the corresponding letter $\lambda '\in Y$. Obviously, $\alpha_R$ is a regular expression over the alphabet
$X\cup Y$. Further, $\|\alpha_R\|$ is considered as a fuzzy
language over an alphabet $X\cup Y$, taking values in the set
$\{0, e\}\subseteq L.$

Let $\varphi_\alpha:X\cup Y\to L$ be a mapping defined by
\begin{equation}\label{eq:ned-fuzzy}
\varphi_\alpha(x)=
\begin{cases}
\ e, & \text{if}\ x\in X\\
\ \lambda, & \text{if}\ x=\lambda '\in Y
\end{cases}\ ,
\end{equation}
for any $x\in X\cup Y$. Denote by $\varphi_\alpha^*$ a homomorphism from the
monoid $(X\cup Y)^*$ into the monoid $(L,\otimes, e)$ defined by:
$\varphi_\alpha^*(\varepsilon)=e$ and
$\varphi_\alpha^*(u)=\varphi_\alpha(x_1)\otimes\varphi_\alpha(x_2)\otimes\cdots\otimes\varphi_\alpha(x_n),$
for any $u=x_1x_2\cdots x_n$ with $x_1,\ldots,x_n\in (X\cup Y)^*$.

\begin{example}\rm\label{pr3}
Let ${\cal L}$ be an arbitrary $\ell-$monoid and let $\alpha$ be a
fuzzy regular expression over an alphabet $X.$ If $\alpha$ is
without scalars then $\alpha_R=\alpha.$
\end{example}
\begin{example}\rm\label{pr4}
Let ${\cal L}$ be the G\" odel structure. Consider
$\alpha=0.2((0.1(xy)^*)^*+y),$ a fuzzy regular expression~over the
alphabet $\{x,y\}$. An expression
$\alpha_R=\lambda((\mu(xy)^*)^*+y)$ is a regular expression over
the alphabet $\{x,y,\lambda, \mu\},$ obtained from $\alpha$ by
replacing $0.2$ with $\lambda$ and $0.1$ with $\mu.$ The mapping
$\varphi_{\alpha}$ is given by
\[
\varphi_{\alpha}=
\begin{pmatrix}
x & y & \lambda & \mu \\
1 & 1 & 0.2 & 0.1
\end{pmatrix}.
\]
\end{example}
\begin{example}\rm\label{pr5}
Consider a fuzzy regular expression $\alpha=(0.1x^*)(yx+0.8y)^*,$
where $\cal L$ is the product structure. Then $\alpha_R=(\lambda
x^*)(yx+\mu y)^*$ is a regular expression over the alphabet
$\{x,y,\lambda,\mu\},$ where $\lambda$  replaces
$0.1$ and $\mu$ replaces $0.8.$ The mapping $\varphi_{\alpha}$ is
given by
\[
\varphi_{\alpha}=\begin{pmatrix}
x & y & \lambda & \mu \\
1 & 1 & 0.1 & 0.8
\end{pmatrix}.
\]
\end{example}

Now we prove the following.

\begin{lemma}\label{lema2}
Let ${\cal L}$ be an integral $\ell-$monoid, and let $X$ be an
arbitrary alphabet. Then every homomorphism $\varphi$ from the monoid
$X^*$ into the monoid $(L,\otimes, 1)$ is antitone, i.e.,
\begin{equation}\label{eq:antitonic}
u\le_{em}v\quad\implies\quad\varphi(v)\le\varphi(u),
\end{equation}for all $u,v\in X^*$. Furthermore, for any $U\subseteq X^*$ and any $\gamma:X^*\to \{0,1\}$ there exists $\bigvee \{\varphi (u)\otimes\gamma
(u)\mid u\in U\}$ and
\begin{equation}\label{eq:UMU}
\bigvee_{u\in U} \varphi (u)\otimes \gamma (u) =\bigvee_{u\in M(U')} \varphi (u)\otimes \gamma(u),
\end{equation}
where $U'=\{u\in U\mid \gamma(u)=1\}$.
\end{lemma}

\begin{proof}
If $u\le_{em} v$, then by (\ref{eq:embedding}) we have
\[
u=u_1u_2\cdots u_n\mbox{\ \ and\ \ }v=v_0u_1v_1u_2\cdots
v_{n-1}u_nv_n,
\]
where $n\in \Bbb N$ and $u,v,u_1,u_2,\dots,u_n,v_0v_1,\dots,v_n\in
X^*.$ Consequently, according to (\ref{isotonost}) ,we have
\[
\begin{aligned}
\varphi(v)&=\varphi(v_0)\otimes\varphi(u_1)\otimes\varphi(v_1)\otimes\varphi(u_2)\otimes\cdots
\otimes\varphi(v_{n-1})\otimes\varphi(u_n)\otimes\varphi(v_n)\\
&\le 1\otimes\varphi(u_1)\otimes 1\otimes\varphi(u_2)\otimes\cdots
\otimes 1\otimes\varphi(u_n)\otimes 1 = \varphi(u_1)\otimes \varphi(u_2)\otimes\cdots \otimes
\varphi(u_n)=\varphi(u).
\end{aligned}
\]
Therefore, $\varphi (v)\le \varphi (u)$.

Further, for any $U\subseteq X^*$, $\gamma :X^*\to \{0,1\}$, and $u\in U'=\{v\in
U\mid \gamma(v)=1\}$ there exists $w\in M(U')$ such~that $w\le_{em} u$, and by (\ref{eq:antitonic}) it follows that $\varphi(u)\le_{em}\varphi (w)$. According to
Proposition \ref{prop1}, we have that $M(U')$ is~finite, and by Lemma \ref{lema1} we obtain that $\bigvee \{\varphi (u)\otimes \gamma(u)\mid u\in U\}=\bigvee \{\varphi (u)\mid u\in U'\}$ exists and (\ref{eq:UMU}) holds.
\end{proof}

In particular, for a given  regular expression $\alpha$, the homomorphism $\varphi^*_\alpha$ satisfies~(\ref{eq:antitonic}) and (\ref{eq:UMU}).

Let $Z$ be an alphabet. The {\it shuffle\/} operation, denoted by
$\shuffle$ is defined in the following way
\begin{equation}\label{eq:shuffle1}
u\shuffle v=\bigl\{u_1v_1u_2v_2\cdots u_nv_n\bigm | u=u_1u_2\dots u_n,
v=v_1v_2\dots v_n, u_i, v_i\in Z^*,1\le i\le n, n\in\Bbb N \bigr\},
\end{equation}
where $u,v\in Z^*$.

The above operation is naturally extended to languages by the {\it
shuffle\/} of languages, defined as
\begin{equation}\label{eq:shuffle3}
L_1\shuffle L_2=\bigcup_{u\in L_1,v\in L_2}u\shuffle v.
\end{equation}where $L_1,L_2\subseteq Z^*$.

Let us return now to the fuzzy regular expression $\alpha $ over the alphabet $X$ and the alphabet $Y$ associated with $\alpha $. Supposing
$\emptyset^*=\{\varepsilon\}$, for any $u\in X^*$ we define a language $U_Y(u)\subseteq (X\cup Y)^*$ by
\[
U_Y(u)=u\shuffle Y^*.
\]
It is easy
to check that the following holds
\begin{gather}
U_Y(\varepsilon)=Y^*,\label{eq:sh1}\\
\text{If}\ Y=\emptyset\ \text{then}\ U_Y(u)=\{
u\},\qquad\text{for every}\ u\in X^* ,\label{eq:sh2}\\
U_Y(u)U_Y(v)=U_Y(uv),\qquad\text{for all}\ u, v\in X^*,\label{eq:sh3}\\
U_Y(x)=Y^*xY^*,\qquad\text{for every}\ x\in X,\label{eq:sh4}
\end{gather}
where the set $U_Y(u)U_Y(v)$ is the concatenation of sets $U_Y(u)$
and $U_Y(v)$, and $Y^*xY^*$ is the concatenation of $Y^*$, $\{
x\}$ and $Y^*$.

One of the main results of this paper is the following theorem.

\begin{theorem}\label{lema3} Let $\cal L$ be an integral $\ell-$monoid. Let $\alpha$ be a fuzzy regular expression over a finite
alphabet $X$, and let $Y$ be an alphabet associated with $\alpha$.
Then
\begin{equation}\label{reg}
\|\alpha\|(u)=\bigvee_{v\in
U_Y(u)}\varphi^*_\alpha(v)\otimes\|\alpha_R\|(v),
\end{equation} for every $u\in X^*.$
\end{theorem}

\begin{proof}
Consider an arbitrary $u\in X^*$.

For $U=U_Y(u)$, if the set $U'=\{v\in U\mid \|\alpha_R\|(v)=1\}$ is non-empty, then by Lemma \ref{lema2} it follows that the supremum~on the right side of (\ref{reg}) exists, and
\begin{equation}\label{eq:sup.ex}
\bigvee_{v\in U_Y(u)}\varphi^*_\alpha(v)\otimes\|\alpha_R\|(v)  =
\bigvee_{v\in M(U')}\varphi^*_\alpha(v)\otimes\|\alpha_R\|(v).
\end{equation}
Otherwise, if $U'=\emptyset $, then
\[
\bigvee_{v\in U_Y(u)}\varphi^*_\alpha(v)\otimes\|\alpha_R\|(v)=0.
\]
Thus, we have proved that the supremum on the right side of (\ref{reg}) always exists.

Further, if $\alpha$ is a fuzzy regular expression without scalar multiplication,
i.e., if $Y=\emptyset,$ then $U_Y(u)=\{u\},$ $\alpha=\alpha_R$,
$\|\alpha\|=\|\alpha_R\|$, and $\varphi^*_\alpha(v)=1$ for every
$v\in X^*=(X\cup Y)^*$. As a result, we have
\[
\bigvee_{v\in
U_Y(u)}\varphi^*_\alpha(v)\otimes\|\alpha_R\|(v)=\|\alpha_R\|(u)=\|\alpha\|(u).
\]

The rest of the proof will be done by induction of the length of the
fuzzy regular expression $\alpha.$ Suppose that~(\ref{reg}) holds for
arbitrary  fuzzy regular expressions whose
length is less than the length of $\alpha$.

Let $\alpha=\lambda\beta,$ for $\lambda\in L$ and $\beta \in \cal{LR}$, and let $Y_1\subseteq Y$ be the alphabet associated with $\beta $.~For each $v\in
(X\cup Y)^*$ we have
\[
\|\alpha_R\|(v)=\begin{cases}
\ \|\beta_R\|(w) & \text{if}\ v=\lambda ' w,\ \text{for some}\ w\in (X\cup Y)^*\\
\ 0 & \text{otherwise}
\end{cases}.
\]
For~every $w\in (X\cup Y_1)^*$ we have that $\varphi_\alpha^*(w)=\varphi_\beta^*(w)$
and $\lambda'w\in U_Y(u)$ if and only if $w\in U_{Y_1}(u)$, and also, for every $w\in (X\cup Y)^*$ which contains a letter from $Y\setminus Y_1$
we have that $\| \beta_R\|(w)=0$.~Consequently,
\[
\begin{aligned}
\bigvee_{v\in
U_Y(u)}\varphi^*_\alpha(v)\otimes\|\alpha_R\|(v)&=\bigvee_{w\in
U_{Y_1}(u)}\varphi^*_\alpha(\lambda
'w)\otimes\|\beta_R\|(w)=^*\lambda\otimes\bigvee_{w\in
U_{Y_1}(u)}\varphi^*_\beta(w)\otimes\|\beta_R\|(w)=\lambda\otimes\|\beta\|(u)=\|\alpha\|(u).
\end{aligned}
\]
The equality  marked with * follows by Lemmas \ref{lema1} and \ref{lema2}.

Let $\alpha=\beta+\gamma$, for $\beta,\gamma\in \cal{LR}$, let $Y_1\subseteq Y$ be the alphabet associated with $\beta $,
and let $Y_2\subseteq Y$ be the alphabet associated with $\gamma $.~For every
$v\in (X\cup Y)^*$ we have that the following is true
\begin{align}
&\|\alpha_R\|(v)=\|\beta_R\|(v)\vee\|\gamma_R\|(v),\notag \\
&\varphi_\alpha^*(v)=\varphi_\beta^*(v),\ \text{for}\ v\in (X\cup Y_{1})^*,\
\text{and}\ \varphi_\alpha^*(v)=\varphi_\gamma^*(v),\ \text{for}\ v\in  (X\cup Y_{2})^*,\label{eq:cond.phi} \\
&\|\beta_R\|(v)=0,\ \text{for}\ v\notin (X\cup Y_{1})^*,\ \text{and}\ \|\gamma_R\|(v)=0,\ \text{for}\ v\notin (X\cup Y_{2})^*.\label{eq:cond.bg}
\end{align}
Therefore,
\[
\begin{aligned}
&\bigvee_{v\in
U_Y(u)}\varphi^*_\alpha(v)\otimes\|\alpha_R\|(v)=\bigvee_{v\in
U_Y(u)}\varphi^*_\alpha(v)\otimes(\|\beta_R\|(v)\vee\|\gamma_R\|(v))\\
&\qquad\qquad=^*\big(\bigvee_{v\in
U_Y(u)}\varphi^*_\alpha(v)\otimes\|\beta_R\|(v)\big)\vee\big(\bigvee_{v\in
U_Y(u)}\varphi^*_\alpha(v)\otimes\|\gamma_R\|(v)\big)\\
&\qquad\qquad=\big(\bigvee_{v\in
U_{Y_1}(u)}\varphi^*_\beta(v)\otimes\|\beta_R\|(v)\big)\vee\big(\bigvee_{v\in
U_{Y_2}(u)}\varphi^*_\gamma(v)\otimes\|\gamma_R\|(v)\big) =\|\beta\|(u)\vee\|\gamma\|(u)=\|\alpha\|(u),
\end{aligned}
\]
The equality marked with * follows by Lemmas \ref{lema1} and \ref{lema2}.

Next, let $\alpha=\beta\gamma$, for $\beta,\gamma\in \cal{LR}$, let $Y_1\subseteq Y$ be the alphabet associated with $\beta $,
and let $Y_2\subseteq Y$ be the alphabet associated with $\gamma $.~Then (\ref{eq:cond.phi}) and (\ref{eq:cond.bg}) hold, and $\|\alpha_R\|(v)=\bigvee_{v=wp}\|\beta_R\|(w)\otimes\|\gamma_R\|(p)$, for every $v\in (X\setminus Y)^*$.
 Thus
\[
\begin{aligned}
\bigvee_{v\in
U_Y(u)}\varphi^*_\alpha(v)\otimes\|\alpha_R\|(v)&=\bigvee_{v\in
U_Y(u)}\varphi^*_\alpha(v)\otimes\bigvee_{v=wp}(\|\beta_R\|(w)\otimes\|\gamma_R\|(p))\\
&=^*\bigvee_{v\in U_Y(u)}\ \bigvee_{v=wp} \big(\varphi^*_\alpha(w)\otimes\|\beta_R\|(w)\big)\otimes\big(\varphi^*_\alpha(p)\otimes\|\gamma_R\|(p)\big)\\
&=\bigvee_{u=qr}\ \bigvee_{w\in U_Y(q),\, p\in
U_Y(r)}\big(\varphi^*_\alpha(w)\otimes\|\beta_R\|(w)\big)\otimes\big(\varphi^*_\alpha(p)\otimes\|\gamma_R\|(p)\big)\\
&=\bigvee_{u=qr}\ \bigvee_{w\in U_{Y_1}(q),\, p\in
U_{Y_2}(r)}\big(\varphi^*_\beta(w)\otimes\|\beta_R\|(w)\big)\otimes\big(\varphi^*_\gamma(p)\otimes\|\gamma_R\|(p)\big)\\
&=^{**} \bigvee_{u=qr}\biggl( \big(\bigvee_{w\in U_{Y_1}(q)}\varphi^*_\beta(w)\otimes\|\beta_R\|(w)\big)\otimes\big(\bigvee_{p\in
U_{Y_2}(r)}\varphi^*_\gamma(p)\otimes\|\gamma_R\|(p)\big)\biggr)\\
&=\bigvee_{u=qr}\|\beta(q)\|\otimes\|\gamma(r)\|=\|\alpha\|(u),
\end{aligned}
\]
The equality marked with * follows by
$\varphi^*_\alpha(p)\otimes\|\beta_R\|(w)=\|\beta_R\|(w)\otimes\varphi^*_\alpha(p)$, which is true
since $\|\beta_R\|(w)\in \{0,1\}$, and the equality marked with ** follows by Lemmas \ref{lema1} and \ref{lema2}.

Finally, let $\alpha=\beta^*$, for $\beta\in \cal{LR}$, and for any $n\in \mathbb N$ let $\beta_n=\varepsilon+\beta+\cdots+\beta^n$.~Clearly, $\beta $ and $\beta_n$ have the same associated alphabet as $\alpha $, the alphabet $Y$.~Also,
$\|\beta_n\|(u)\le\|\alpha\|(u),$ for all $u\in X^*$~and $n\in
\Bbb N$.~In addition, by the proof of Proposition \ref{prop3} (cf.~\cite[p.~80]{LP.05}),
we have that for every $u\in X^*$ there exists~$n\in \Bbb N$ such that $\|\alpha\|(u)\le\|\beta_n\|(u)$, and then
\[
\|\alpha\|(u)\le\|\beta_n\|(u)=\bigvee_{v\in
U_Y(u)}\varphi^*_\alpha(v)\otimes\|(\beta_n)_R\|(v)\le
\bigvee_{v\in U_Y(u)}\varphi^*_\alpha(v)\otimes\|\alpha_R\|(v),
\]
for every $u\in X^*$.~Conversely, for every $u\in X^*$ and $v\in U_Y(u)$ there exists $m\in \Bbb
N$ such that
\[
\varphi^*_\alpha(v)\otimes\|\alpha_R\|(v)\le\varphi^*_\alpha(v)\otimes\|(\beta_m)_R\|(v)\le\|\beta_m\|(u)\le\|\alpha\|(u).
\]
In conclusion,
\[
\|\alpha\|(u)=\bigvee_{v\in U_Y(u)}\varphi^*_\alpha(v)\otimes\|\alpha_R\|(v),
\]
which completes the proof of the theorem.
\end{proof}

For a fuzzy regular expression $\alpha$ over an alphabet $X,$ let
$\alpha_R$ be a regular expression over an alphabet $X\cup Y,$
where $Y$ is an alphabet associated with $\alpha$. Now, let ${\cal
A}=(A, X\cup Y, \delta^A, a_0, \tau^A)$ be an arbitrary
nondeterministic automaton recognizing the language
$\|\alpha_R\|.$ Evidently, the automaton $\cal A$, considered as a fuzzy automaton,
recognizes the fuzzy language $\|\alpha_R\|.$ Further, let ${\cal
A}_\alpha=(A_\alpha, X, \delta^{A_\alpha}, a_0^\alpha, \tau^{A_\alpha})$
be a fuzzy automaton with   $A_\alpha=A$, $a_0^{\alpha}=a_0$,
 and a fuzzy transition relation $\delta^{A_\alpha}$
defined by
\begin{equation}\label{eq:fuzzyrec.1}
\delta^{A_\alpha}(a, x, b)=\bigvee_{v\in
U_Y(x)}\varphi_\alpha^*(v)\otimes\delta^A(a, v, b),
\end{equation}
for all $a,b\in A_\alpha$ and $x\in X$,
\begin{equation}\label{eq:fuzzyrec.2}
\tau^{A_\alpha}(a)=\bigvee_{v\in Y^*}\bigvee_{b\in
A}\varphi_\alpha^*(v)\otimes\delta^A(a, v,
b)\otimes\tau^A(b),
\end{equation}
or equivalently,
\[
\tau^{A_\alpha}(a)=\bigvee_{v\in
Y^*}\varphi_\alpha^*(v)\otimes(\delta^A_v\circ\tau^A)(a).
\]
for each $a\in A_\alpha$.~Note that the existence of the above suprema by
$v\in U_Y(x)$ and $v\in Y^*$ follows immediately by equation (\ref{eq:UMU})
in Lemma \ref{lema2}.

We prove the following fundamental result.

\begin{theorem}\label{th:1}
Let $\cal L$ be an integral $\ell-$monoid, let $\alpha $ be a fuzzy regular expression, and let ${\cal A}=(A,X\cup Y, \delta^A, a_0,\tau^A)$ be an arbitrary nondeterministic
automaton which recognizes $\|\alpha_R\|$.

Then ${\cal A}_\alpha=(A_\alpha, X, \delta^{A_\alpha}, a_0, \tau^{A_\alpha})$ is a well-defined fuzzy automaton and it recognizes the fuzzy language $\|\alpha\|$.
\end{theorem}

\begin{proof}
According to (\ref{languageinit}), we have
\[
L({\cal A}_\alpha)(u)=\bigvee_{a\in A_\alpha}\delta^{A_\alpha}(a_0, u,
a)\otimes\tau^{A_\alpha}(a).
\]
Thus, for the empty word
$\varepsilon\in X^*$, by Theorem \ref{lema3}, we have
\[
\begin{aligned}
L({\cal
A}_{\alpha})(\varepsilon)&=\tau^{A_\alpha}(a_0)=\bigvee_{v\in
Y^*}\bigvee_{b\in A}\varphi_\alpha^*(v)\otimes\delta^A(a_0, v,
b)\otimes\tau^A(b)=\bigvee_{v\in Y^*}\varphi_\alpha^*(v)\otimes\bigl(\bigvee_{b\in
A}\delta^A(a_0, v, b)\otimes\tau^A(b)\bigr)\\
&=\bigvee_{v\in
U_Y(\varepsilon)}\varphi_\alpha^*(v)\otimes\|\alpha_R\|(v)=\|\alpha\|(\varepsilon).
\end{aligned}
\]

Suppose that $\delta^{A_\alpha}(a,u,b)=\bigvee\{\,\varphi^*_\alpha(v)\otimes\delta^A(a,v,b)\mid
{v\in U_Y(u)}\,\}$, for some $u\in X^*$
and all $a,b\in A_\alpha.$ Then for any $x\in X$ we have
\[
\begin{aligned}
\delta^{A_\alpha}(a,ux,b) &=\bigvee_{c\in
A_\alpha}\delta^{A_\alpha}(a,u,c)\otimes\delta^{A_\alpha}(c,x,b)=\bigvee_{c\in A}\bigl(\bigvee_{v\in
U_Y(u)}\varphi^*_\alpha(v)\otimes\delta^A(a,v,c)\bigr)\otimes\bigl(\bigvee_{w\in
U_Y(x)}\varphi^*_\alpha(w)\otimes\delta^A(c,w,b)\bigr)\\
&=\bigvee_{c\in A}\bigvee_{v\in U_Y(u),\atop w\in
U_Y(x)}\varphi^*_\alpha(v)\otimes\delta^A(a,v,c)\otimes\varphi^*_\alpha(w)\otimes\delta^A(c,w,b)=\bigvee_{v\in U_Y(u),\atop w\in
U_Y(x)}\varphi^*_\alpha(vw)\otimes\bigl(\bigvee_{c\in
A}\delta^A(a,v,c)\otimes\delta^A(c,w,b)\bigr)\\
&=\bigvee_{v\in U_Y(u),\atop w\in
U_Y(x)}\varphi^*_\alpha(vw)\otimes\delta^A(a,vw,b)=\bigvee_{v\in
U_Y(ux)}\varphi^*_\alpha(v)\otimes\delta^A(a,v,b).
\end{aligned}
\]
Observe that the above equalities follow by Lemmas \ref{lema1} and \ref{lema2}, and equation (\ref{eq:sh3}).~We have also used the equality $\delta^A(a,v,c)\otimes\varphi^*_\alpha(w)=\varphi^*_\alpha(w)\otimes\delta^A(a,v,c)$,
which follows by the fact that $\delta^A(a,v,c)\in \{0,1\}$.

Consequently, for any $u\in X^+$, due to Theorem \ref{lema3},
(\ref{eq:fuzzyrec.1}) and (\ref{eq:fuzzyrec.2}), we have
\[
\begin{aligned}
L({\cal A}_\alpha)(u)&=\bigvee_{a\in
A_\alpha}\delta^{A_\alpha}(a_0,u,a)\otimes\tau^{A_\alpha}(a)=\bigvee_{a\in A}\big(\bigvee_{v\in
U_Y(u)}\varphi^*_\alpha(v)\otimes\delta^A(a_0,v,a)\big)\otimes\big(\bigvee_{w\in
Y^*}\bigvee_{b\in A}\varphi_\alpha^*(w)\otimes\delta^A(a, w,
b)\otimes\tau^A(b)\big)\\
&=\bigvee_{v\in U_Y(u)\atop w\in
Y^*}\varphi_\alpha^*(vw)\otimes\big(\bigvee_{a, b\in
A}\delta^A(a_0,v,a)\otimes\delta^A(a, w,
b)\otimes\tau^A(b)\big)=\bigvee_{v\in U_Y(u)}\varphi_\alpha^*(v)\otimes\bigl(\bigvee_{a\in
A}\delta^A(a_0,v,b)\otimes\tau^A(b)\bigr)\\
&=\bigvee_{v\in
U_Y(u)}\varphi^*_\alpha(v)\otimes\|\alpha_R\|(v)=\|\alpha\|(u).
\end{aligned}
\]
This completes the proof of the theorem.
\end{proof}

The fuzzy automaton ${\cal A}_\alpha=(A_\alpha, X, \delta^{A_\alpha}, a_0, \tau^{A_\alpha})$ will be called the {\it fuzzy automaton associated
with $\cal A$ and $\alpha$\/}.

\section{Fuzzy automata from fuzzy regular expressions: Effective construction}\label{sec:effect.con}

Let $\cal L$ be an integral $\ell-$monoid, and let $\alpha$ be an
arbitrary  fuzzy regular expression. Theorem
\ref{th:1} allows us to construct different types of  fuzzyautomata from $\alpha$, i.e., different  fuzzy automata recognizing the fuzzy language $\|\alpha\|$.~Namely, in the
general case, by choosing different nondeterministic automata $\cal A$ constructed
  from $\alpha_R$, we obtain different  fuzzy
finite automata ${\cal A}_\alpha$ of $\alpha$.

Let us recall that there are many well-known constructions of
small nondeterministic automata from a given regular
expression. The most famous are those of Thompson \cite{T.68},
Glushkov \cite{G.61} and McNaughton-Yamada \cite{MY.60}.~The last
one is known as the {\it position automaton\/}. In addition, Antimirov in
\cite{A.96}  constructed~the {\it partial derivative automaton\/}, which
generalizes Brzozowski's {\it derivative automaton\/} \cite{B.64}.~However, in spite of improve\-ments made by Brzozowski and
Antimirov, the position automaton is the most often~used, probably
because of its simplicity and the fact that other constructions
did not make any practical improvements.~The latest
nondeterministic automaton constructed from a regular exporession is the
{\it follow automaton\/}, introduced by Ilie et al.~\cite{IY.02a,IY.02,IY.03a}.~It
has been proved that the follow automaton is the quotient of the
position auto\-maton, and therefore it is smaller than the
position automaton.

Let $\cal L$ be an integral $\ell-$monoid,~and let $\alpha$ be an
arbitrary  fuzzy regular expression over an alphabet~$X$. Consider a
regular expression $\alpha_R$ over $X\cup Y$, where $Y$ is an
alphabet associated with $\alpha$.~Starting from~the position automaton
${\cal A}_{\rm p}(\alpha_R)=(A_{\rm p}, X\cup Y,
\delta^{A_{\rm p}}, 0, \tau^{A_{\rm p}})$ of $\alpha_R$, by means of (\ref{eq:fuzzyrec.1}) and (\ref{eq:fuzzyrec.2})
we construct the~fuzzy automaton associated with~${\cal A}_{\rm p}(\alpha_R)$~and $\alpha$, which will be denoted by ${\cal
A}_{\rm pf}(\alpha)=(A_{\rm pf}, X, \delta^{A_{\rm pf}}, 0,
\tau^{A_{\rm pf}})$.~The~computing of ${\cal A}_{\rm
pf}(\alpha)$ for a given  fuzzy regular expression is described in
Theorem \ref{th:1}, and Examples \ref{pr6} and \ref{pr7} clarify this construction.

\begin{example}\rm\label{pr6}
Let $\cal L$ be the G\" odel structure.~Consider
$\alpha=0.2((0.1(xy)^*)^*+y),$ a fuzzy regular expression over the
alphabet $\{x,y\}$ from Example \ref{pr4}.~Here
$\alpha_R=\lambda((\mu(xy)^*)^*+y)$ is a regular expression over
the alphabet $\{x,y,\lambda, \mu\},$ obtained from $\alpha$. The
marked version of the expression $\alpha_R$ is
$\overline{\alpha_R}=\lambda_1((\mu_2(x_3y_4)^*)^*+y_5),$ and
$\varphi_{\alpha}$ is given by
\[\varphi_{\alpha}=\begin{pmatrix}
x & y & \lambda & \mu \\
1 & 1 & 0.2 & 0.1
\end{pmatrix}.
\]
The picture bellow represents the graph of the position
automaton ${\cal A}_{\rm p}(\alpha_R)$:
\begin{center}
\psset{unit=1.2cm}
\newpsobject{showgrid}{psgrid}{subgriddiv=1,griddots=10,gridlabels=6pt}
\begin{pspicture}(0,-2)(6,2)

\pnode(0,0){I} \SpecialCoor
\rput([angle=0,nodesep=8mm,offset=0pt]I ){\cnode{3mm}{A0}}
\rput([angle=0,nodesep=15mm,offset=0pt]A0){\cnode{3mm}{A1}}
\rput([angle=0,nodesep=15mm,offset=0pt]A0){\cnode{2.5mm}{A1}}
\rput([angle=0,nodesep=15mm,offset=0pt]A1){\cnode{3mm}{A2}}
\rput([angle=0,nodesep=15mm,offset=0pt]A1){\cnode{2.5mm}{A2}}
\rput([angle=0,nodesep=15mm,offset=0pt]A2){\cnode{3mm}{A4}}
\rput([angle=0,nodesep=15mm,offset=0pt]A2){\cnode{2.5mm}{A4}}
\rput([angle=-90,nodesep=15mm,offset=0pt]A2){\cnode{3mm}{A3}}
\rput([angle=90,nodesep=15mm,offset=0pt]A1){\cnode{3mm}{A5}}
\rput([angle=90,nodesep=15mm,offset=0pt]A1){\cnode{2.5mm}{A5}}
\rput(A0){\footnotesize $0$} \rput(A1){\footnotesize $1$}
\rput(A2){\footnotesize $2$} \rput(A3){\footnotesize $3$}
\rput(A4){\footnotesize $4$} \rput(A5){\footnotesize $5$}
\NormalCoor \ncline{->}{I}{A0}
\ncline{->}{A0}{A1}\Aput[1pt]{\footnotesize $\lambda$}
\ncline{->}{A1}{A5}\Aput[1pt]{\footnotesize $y$}
\ncline{->}{A1}{A2}\Aput[1pt]{\footnotesize $\mu$}
\ncline{->}{A2}{A3}\Aput[1pt]{\footnotesize $x$}
\ncline{->}{A4}{A2}\Bput[1pt]{\footnotesize $\mu$}

\ncarc[arcangle=15]{->}{A3}{A4}\aput[0pt](.5){\footnotesize $y$}
\ncarc[arcangle=15]{->}{A4}{A3}\aput[0pt](.4){\footnotesize $x$}

\nccircle[angleA=0]{<-}{A2}{0.3}\Bput[1pt]{\footnotesize $\mu$}

\end{pspicture}\\
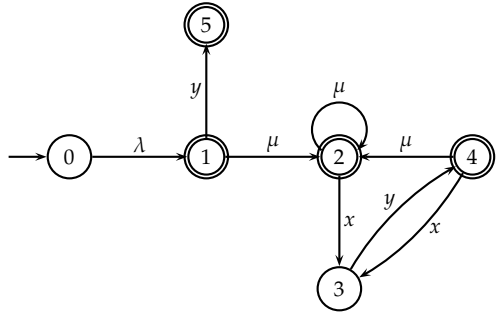
Figure 1. The automaton  ${\cal A}_{\rm p}(\alpha_R)$
\end{center}
\medskip

Let us observe that
\[
\delta^{A_{\rm pf}}(i,x,j)=
\begin{cases}
\ \displaystyle\bigvee_{u\in {\cal M}(i,x,j)}\varphi^*_\alpha(u) & \text{if}\
{\cal P}(i,x,j)\not=\emptyset\\
\quad\ \ \,0 & \text{otherwise}
\end{cases},
\]
for all $x\in X$, $i,j\in A_{\rm p}$ (in notation from the proof of Theorem \ref{th:1}).

Let us, for example, describe how to determine $\delta^{A_{\rm
pf}}(0,x,3).$ For each word $u\in {\cal M}(0,x,3)$ there is a path
in the graph of ${\cal A}_{\rm p}(\alpha_R),$ which starts in $0$
and ends in $3$, with a single edge marked with $x$ and with all
other edges marked with symbol $\lambda$ or $\mu$ (see Figure 1.).
Obviously, ${\cal M}(0,x,3)=\{\lambda\mu x\}.$ Now,
\[
\delta^{A_{\rm pf}}(0,x,3)=\varphi^*_\alpha(\lambda\mu
x)=0.2\otimes0.1\otimes 1=0.1
\]
Further, ${\cal M}(1,x,3)=\{\mu x\},$ ${\cal M}(2,x,3)=\{x\},$
${\cal M}(4,x,3)=\{x\}$ and ${\cal M}(i,x,j)=\emptyset$ in all
other cases, and we have
\[
\begin{aligned}
\delta^{A_{\rm pf}}(1,x,3)&=0.1,\ \ \delta^{A_{\rm pf}}(2,x,3)=1,\\
\delta^{A_{\rm pf}}(4,x,3)&=1, \mbox{\ and\ \ }\delta^{A_{\rm
pf}}(i,x,j)=0,
\end{aligned}
\]
for $(i,j)\not\in \{(0,3),(1,3),(2,3),(4,3)\}.$

From Figure 1 one can see that ${\cal M}(0,y,5)=\{\lambda y\}$,
${\cal M}(1,y,5)=\{ y\}$, ${\cal M}(3,y,2)=\{ y\mu\}$ and ${\cal
M}(3,y,4)=\{ y\}$, whereas ${\cal M}(i,y,j)=\emptyset$ in all other
cases. Therefore, we have
\[
\delta^{A_{\rm pf}}(0,y,5)=0.2,\ \ \delta^{A_{\rm pf}}(1,y,5)=1,\ \
\delta^{A_{\rm pf}}(3,y,2)=0.1,\ \ \delta^{A_{\rm pf}}(3,y,4)=1,\ \
\delta^{A_{\rm pf}}(i,y,j)=0,
\]
for $(i,j)\not\in\{(0,5),(1,5)(3,2),(3,4)\}.$

To summarize,  fuzzy transition relations $\delta^{A_{\rm pf}}_x$,
$\delta^{A_{\rm pf}}_y$, and the  fuzzy set $\tau^{A_{\rm pf}}$ of
final states of the  fuzzy automaton ${\cal A}_{\rm
pf}(\alpha)$ are:
\[\delta^{A_{\rm
pf}}_x=\begin{bmatrix}
                   0 & 0 & 0 & 0.1 & 0 & 0\\
                   0 & 0 & 0 & 0.1 & 0 & 0\\
                   0 & 0 & 0 & 1 & 0 & 0\\
                   0 & 0 & 0 & 0 & 0 & 0\\
                   0 & 0 & 0 & 1 & 0 & 0\\
                   0 & 0 & 0 & 0 & 0 & 0
                   \end{bmatrix},\qquad
\delta^{A_{\rm pf}}_y=\begin{bmatrix}
                   0 & 0 & 0 & 0 & 0 & 0.2\\
                   0 & 0 & 0 & 0 & 0 & 1\\
                   0 & 0 & 0 & 0 & 0 & 0\\
                   0 & 0 & 0.1 & 0 & 1 & 0\\
                   0 & 0 & 0 & 0 & 0 & 0\\
                   0 & 0 & 0 & 0 & 0 & 0
                   \end{bmatrix},\qquad
\tau^{A_{\rm pf}}=\begin{bmatrix}
                   0.2 \\
                   1 \\
                   1 \\
                   0 \\
                   1 \\
                   1
                   \end{bmatrix},
\]and the graph of ${\cal A}_{\rm pf}(\alpha)$ is presented by Figure
2.
\begin{center}
\psset{unit=1.2cm}
\newpsobject{showgrid}{psgrid}{subgriddiv=1,griddots=10,gridlabels=6pt}
\begin{pspicture}(0,-1.8)(6,2)

\pnode(0,0){I} \SpecialCoor
\rput([angle=0,nodesep=8mm,offset=0pt]I ){\cnode{3mm}{A0}}
\rput([angle=0,nodesep=15mm,offset=0pt]A0){\cnode{3mm}{A1}}
\rput([angle=0,nodesep=15mm,offset=0pt]A1){\cnode{3mm}{A3}}
\rput([angle=45,nodesep=22.5mm,offset=0pt]A3){\cnode{3mm}{A2}}
\rput([angle=-90,nodesep=30mm,offset=0pt]A2){\cnode{3mm}{A4}}
\rput([angle=90,nodesep=15mm,offset=0pt]A1){\cnode{3mm}{A5}}
\rput(A0){\footnotesize $0$} \rput(A1){\footnotesize $1$}
\rput(A2){\footnotesize $2$} \rput(A3){\footnotesize $3$}
\rput(A4){\footnotesize $4$} \rput(A5){\footnotesize $5$}
\NormalCoor \ncline{->}{I}{A0}
\ncline{->}{A0}{A5}\Aput[1pt]{\footnotesize $y/0.2$}
\ncline{->}{A1}{A5}\Bput[1pt]{\footnotesize $y/1$}
\ncline{->}{A1}{A3}\Aput[1pt]{\footnotesize $x/0.1$}
\ncarc[arcangle=-35]{->}{A0}{A3}\Bput[1pt]{\footnotesize $x/0.1$}
\ncarc[arcangle=15]{->}{A2}{A3}\Aput[1pt]{\footnotesize $x/1$}
\ncarc[arcangle=15]{->}{A3}{A2}\Aput[1pt]{\footnotesize $y/0.1$}
\ncarc[arcangle=15]{->}{A3}{A4}\aput[0pt](.5){\footnotesize $y/1$}
\ncarc[arcangle=15]{->}{A4}{A3}\aput[0pt](.4){\footnotesize $x/1$}
\end{pspicture}\\
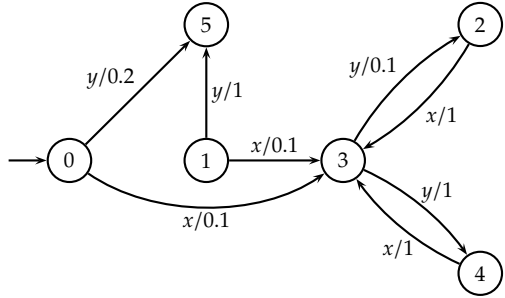
Figure 2. ${\cal A}_{\rm pf}(\alpha)$
\end{center}
\end{example}

It is important to note that the computing of the transition
relation of the  fuzzy automaton ${\cal A}_{\rm
pf}(\alpha)$, for a given  regular expression, might be a problem.~Namely, in the general case, for given $i,j\in A_{\rm pf}$ and $x\in
X$ the set ${\cal P}(i,x,j)$ of all words $u\in U_Y(x)$ such that $\delta^{A_{\rm p}}(i,u,j)=1$ is infinite, and hence, the computing of the~set ${\cal M}(i,x,j)$ of all minimal words of ${\cal
P}(i,x,j)$ with respect to $\le_{em}$ might be a difficult task.~In the sequel we consider this problem.

The next lemma is the well-known result which, for instance, was proved in \cite{DeB.DeM.03} for fuzzy relations with membership values in the real
unit interval and the composition defined by means of a $t$-norm.~In the
same way it can be proved for fuzzy relations over an integral $\ell-$monoid.

\begin{lemma}\label{lema4} Let $\cal L$ be an integral $\ell-$monoid and
let $R$ be a  fuzzy relation on a finite set $A$ with
$|A|=n$. Then
\begin{equation}\label{eq:closure}
\bigvee_{k=1}^n R^k
\end{equation}
is the least transitive  fuzzy relation on $A$ which contains $R.$

In particular, if $R$ is reflexive, then the least
 transitive  fuzzy relation on $A$ containing $R$ is equal to $R^n$.
\end{lemma}

Let $\cal L$ be an integral $\ell-$monoid, and let $\alpha$ be an
arbitrary  fuzzy regular expression over an alphabet $X$. For a
regular expression $\alpha_R$ over $X\cup Y$, where $Y$ is an
alphabet associated with $\alpha,$ and let  ${\cal A}=(A, X\cup Y,
\delta^A, a_0, \tau^A)$ be an arbitrary nondeterministic automaton which recognizes the language $\|\alpha_R\|$.~Let us define
a reflexive  fuzzy relation $R$ on $A$ as follows
\begin{equation}\label{eq:closure}
R(a,b)=\begin{cases}
       \ \qquad\quad1 & \text{if}\ a=b\\
       \ \displaystyle\bigvee_{\lambda '\in Y}\lambda\otimes\delta^A(a,\lambda ',b) & \text{otherwise}
       \end{cases},
\end{equation}
and let us denote by $R_{\cal A}$ the least transitive relation containing $R$.~By Lemma \ref{lema4} we obtain that $R_{\cal
A}=R^n$, where $n$ is the number of states of $\cal A$.
Now, we can prove the following:

\begin{theorem}\label{lemanova1}
Let $\cal L$ be an integral $\ell-$monoid and let $\alpha$ be an
arbitrary  fuzzy regular expression.~For an arbitrary~non\-deterministic automaton
${\cal A}=(A, X\cup Y,\delta^A, a_0, \tau^A)$ recognizing the language $\|\alpha_R\|$
we have
\begin{equation}\label{eq:pom}
R_{\cal A}(a,b)=\bigvee_{u\in
Y^*}\varphi^*_\alpha(u)\otimes\delta^{A}(a,u,b),
\end{equation}for all $a,b\in A.$
\end{theorem}

\begin{proof}
First we note that the existence of the supremum on the right side of
(\ref{eq:pom}) follows by Lemma \ref{lema2}.

Let $n$ be the number of states of $\cal A$.~If $a=b$ then
\[
R_{\cal
A}(a,a)=1=\varphi^*_\alpha(\varepsilon)=\varphi^*_\alpha(\varepsilon)\otimes\delta^{A}(a,\varepsilon,a)\le\bigvee_{u\in
Y^*}\varphi^*_\alpha(u)\otimes\delta^{A}(a,u,a),
\]
and hence, (\ref{eq:pom}) holds.~Otherwise, we have the following
\[
\begin{aligned}
R_{\cal A}(a,b)&=R^n(a,b)=\bigvee_{a_1,\dots,a_{n-1}\in A}R(a,a_1)\otimes
R(a_1,a_2)\otimes\cdots\otimes R(a_{n-1},b)\\
&=\bigvee_{a_1,\dots,a_{k-1}\in A,\atop k\le n,\, a_i\not=
a_{i+1}}R(a,a_1)\otimes
R(a_1,a_2)\otimes\cdots\otimes R(a_{k-1},b)\\
&=\bigvee_{a_1,\dots,a_{k-1}\in A,\atop k\le n,\, a_i\not=
a_{i+1}}\bigvee_{\lambda_1 ',\dots,\lambda_{k} '\in
Y}\lambda_1\otimes\cdots\otimes\lambda_{k}\otimes\delta^A(a,\lambda_1',a_1)\otimes\cdots\otimes\delta^A(a_{k-1},\lambda_{k} ',b)\\
&=\bigvee_{\lambda_1 ',\dots,\lambda_{k} '\in
Y}\lambda_1\otimes\cdots\otimes\lambda_{k}\otimes\biggl(\bigvee_{a_1,\dots,a_{k-1}\in A,\atop k\le n,\, a_i\not=
a_{i+1}}\delta^A(a,\lambda_1',a_1)\otimes\cdots\otimes\delta^A(a_{k-1},\lambda_{k} ',b)\biggr)\\
&\le\bigvee_{\lambda_1 ',\dots,\lambda_{k} '\in
Y}\lambda_1\otimes\cdots\otimes\lambda_{k}\otimes\delta^A(a,\lambda_1 '\cdots\lambda_{k} ',b)=\bigvee_{\lambda_1 ',\dots,\lambda_{k} '\in
Y}\varphi^*_\alpha(\lambda_1'\cdots\lambda_{k}')\otimes\delta^A(a,\lambda_1 '\cdots\lambda_{k} ',b)\\
&\le\bigvee_{u\in Y^*}\varphi^*_\alpha(u)\otimes\delta^{A}(a,u,b).
\end{aligned}
\]
On the other hand, consider an arbitrary $u=\lambda_1'\cdots \lambda_k'\in
Y^*$, with $\lambda_1',\ldots ,\lambda_k'\in Y$, $k\in \mathbb N$. Then
\[
\begin{aligned}
\varphi_\alpha^*(u)\otimes \delta^A(a,u,b)&=\lambda_1\otimes \cdots \lambda_k\otimes
\delta^A(a,\lambda_1'\cdots \lambda_k',b)\\
&=\lambda_1\otimes \cdots \lambda_k\otimes
\bigvee_{a_1,\ldots, a_{k-1}\in A}\delta^A(a,\lambda_1',a_1)\otimes \cdots\otimes \delta^A(a_{k-1},\lambda_k',b) \\
&\le
\bigvee_{a_1,\ldots, a_{k-1}\in A}\lambda_1\otimes \cdots \lambda_k\otimes\delta^A(a,\lambda_1',a_1)\otimes \cdots\otimes \delta^A(a_{k-1},\lambda_k',b) \\
&\le \bigvee_{a_1,\ldots, a_{k-1}\in A}R(a,a_1)\otimes \cdots \otimes R(a_{k-1},b)=R^k(a,b)\le
R_{\cal A}(a,b),
\end{aligned}
\]
whence it follows that
\[
\bigvee_{u\in Y^*}\varphi_\alpha^*(u)\otimes \delta^A(a,u,b)\le R_{\cal A}(a,b).
\]
Therefore, (\ref{eq:pom}) holds.
\end{proof}

By the previous theorem we can conclude that $R_{\cal A}$ is the
transitive closure of the adjacency matrix of the weighted graph
obtained from the graph of the automaton $\cal A$ by removing all
the edges marked by the symbols from  the alphabet $X$, in which
the weight of the edge marked by $\lambda '\in Y$ equals
$\varphi_\alpha^*(\lambda ')$.

Next we prove the following.

\begin{theorem}\label{th:2}
Let $\cal L$ be an integral $\ell-$monoid and let $\alpha$ be an
arbitrary  fuzzy regular expression.~For an arbitrary~non\-deterministic automaton ${\cal A}=(A, X\cup Y, \delta^A, a_0, \tau^A)$ recognizing the language $\|\alpha_R\|$ and the fuzzy automaton ${\cal
A}_\alpha$ associated with $\cal A$ and $\alpha$ we have
\begin{equation}\label{eq:comp}
\delta^{A_{\alpha}}_x=R_{\cal A}\circ\delta^{A}_x\circ R_{\cal
A},\text{\ \ and\ \ }\tau^{A_{\alpha}}=R_{\cal A}\circ\tau^{A},
\end{equation}
for every $x\in X$.
\end{theorem}
\begin{proof}
By (\ref{eq:fuzzyrec.1}), (\ref{eq:sh4}) and (\ref{eq:pom}), we have
\[
\begin{aligned}
\delta^{A_{\alpha}}(a,x,b)&=\bigvee_{u\in
U_Y(x)}\varphi^*_\alpha(u)\otimes\delta^{A}(a,u,b)=\bigvee_{u,v\in
Y^*}\varphi^*_\alpha(u)\otimes\delta^{A}(a,uxv,b)\otimes\varphi^*_\alpha(v)\\
&=\bigvee_{u,v\in Y^*}\bigvee_{c,d\in A}\varphi^*_\alpha(u)\otimes\delta^{A}(a,u,c)\otimes\delta^{A}(c,x,d)\otimes\delta^{A}(d,v,b)\otimes\varphi^*_\alpha(v)\\
&=\bigvee_{c,d\in A}\bigl(\bigvee_{u\in Y^*}\varphi^*_\alpha(u)\otimes\delta^{A}(a,u,c)\bigr)\otimes\delta^{A}(c,x,d)\otimes\bigl(\bigvee_{v\in Y^*}\varphi^*_\alpha(v)\otimes\delta^{A}(d,v,b)\bigr)\\
&=\bigvee_{c,d\in A}R_{\cal A}(a,c)\otimes\delta^{A}(c,x,d)\otimes
R_{\cal A}(d,b),
\end{aligned}
\]
for all $a,b\in A$.~Let us note that the existence of the above suprema follows
by Lemmas \ref{lema1} and \ref{lema2}.

The rest of the proof follows immediately
from (\ref{eq:fuzzyrec.2}) and Theorem \ref{lemanova1}.
\end{proof}

The previous theorem gives an efficient method for computing the
 fuzzy automaton corresponding~to a given  fuzzy
regular expression $\alpha.$ Namely, for $\alpha$ and a
nondeterministic automaton  $\cal A$ recognizing the language  $\|\alpha_R\|$,  the fuzzy
transition relations of ${\cal A}_{\alpha}$ are just matrix
products of $R_{\cal A}$ and the related  fuzzy transition~rela\-tion of
${\cal A}$ (cf.~Example \ref{pr7}).

\begin{example}\rm\label{pr7}
Consider $\alpha=0.2((0.1(xy)^*)^*+y),$ the  fuzzy regular
expression from Example \ref{pr6}. It is easy to verify, using
Figure 2, that  fuzzy relations $R$ and $R_{{\cal A}_{\rm p}}$ are
those given by matrices\[
R=\begin{bmatrix}
                   1 & 0.2 & 0 & 0 & 0 & 0\\
                   0 & 1 & 0.1 & 0 & 0 & 0\\
                   0 & 0 & 1 & 0 & 0 & 0\\
                   0 & 0 & 0 & 1 & 0 & 0\\
                   0 & 0 & 0.1 & 0 & 1 & 0\\
                   0 & 0 & 0 & 0 & 0 & 1
                   \end{bmatrix},
\qquad R_{{\cal A}_{\rm p}}=\begin{bmatrix}
                   1 & 0.2 & 0.1 & 0 & 0 & 0\\
                   0 & 1 & 0.1 & 0 & 0 & 0\\
                   0 & 0 & 1 & 0 & 0 & 0\\
                   0 & 0 & 0 & 1 & 0 & 0\\
                   0 & 0 & 0.1 & 0 & 1 & 0\\
                   0 & 0 & 0 & 0 & 0 & 1
                   \end{bmatrix}.
\]Now, by Theorem \ref{th:2}, we compute  $\delta^{A_{\rm pf}}_x$, $\delta^{A_{\rm
pf}}_y$ and $\tau^{A_{\rm pf}}$ as follows:
\[
\begin{aligned}
\delta^{A_{\rm pf}}_x&=R_{{\cal A}_{\rm p}}\circ\delta^{A_{\rm
p}}_x\circ R_{{\cal A}_{\rm p}}=
                   \begin{bmatrix}
                   0 & 0 & 0 & 0.1 & 0 & 0\\
                   0 & 0 & 0 & 0.1 & 0 & 0\\
                   0 & 0 & 0 & 1 & 0 & 0\\
                   0 & 0 & 0 & 0 & 0 & 0\\
                   0 & 0 & 0 & 1 & 0 & 0\\
                   0 & 0 & 0 & 0 & 0 & 0
                   \end{bmatrix},\\
\delta^{A_{\rm pf}}_y&=R_{{\cal A}_{\rm p}}\circ\delta^{A_{\rm
p}}_y\circ R_{{\cal A}_{\rm p}}=
                   \begin{bmatrix}
                   0 & 0 & 0 & 0 & 0 & 0.2\\
                   0 & 0 & 0 & 0 & 0 & 1\\
                   0 & 0 & 0 & 0 & 0 & 0\\
                   0 & 0 & 0.1 & 0 & 1 & 0\\
                   0 & 0 & 0 & 0 & 0 & 0\\
                   0 & 0 & 0 & 0 & 0 & 0
                   \end{bmatrix},\ \ \
\tau^{A_{\rm pf}}&=R_{{\cal A}_{\rm p}}\circ\tau^{A_{\rm p}}=
                   \begin{bmatrix}
                   0.2\\
                   1\\
                   1\\
                   0\\
                   1\\
                   1
                   \end{bmatrix}.
\end{aligned}
\]
\end{example}

\section{Fuzzy automata from fuzzy regular expressions: Reduced construction}\label{sec:reduced}

Let $\cal L$ be an integral $\ell-$monoid, and let $\alpha$ be an
arbitrary  fuzzy regular expression over an alphabet $X$. For a
regular expression $\alpha_R$ over $X\cup Y$, where $Y$ is an
alphabet associated with $\alpha,$ let  ${\cal A}=(A, X\cup Y,
\delta^{A}, a_0, \tau^{A})$ be a nondeterministic automaton recognizing the
language
$\|\alpha_R\|$.~Besides, let ${\cal A}_\alpha=(A_\alpha, X,
\delta^{A_\alpha}, a_0, \tau^{A_\alpha})$ be the fuzzy automaton associated with $\cal A$ and $\alpha$. Set
\[
A^{\rm r}_\alpha=\{a_0\}\cup\{a\in A_\alpha\ |\ (\exists b\in
A_\alpha)(\exists x\in X)\ \delta^A(b,x,a)=1\}.
\]
Let us denote by ${\cal
A}^{\rm r}_\alpha=(A^{\rm r}_\alpha, X, \delta^{A^{\rm r}_\alpha},
a_0, \tau^{A^{\rm r}_\alpha})$ a fuzzy automaton defined by
\begin{equation}\label{eq:reduced}
\delta^{A^{\rm r}_\alpha}_x(a,b)=(R_{\cal
A}\circ\delta^{A}_x)(a,b),\quad\tau^{A^{\rm r}_\alpha}(a)=(R_{\cal
A}\circ\tau^{A})(a).
\end{equation}
for all $a,b\in A^{\rm r}_\alpha,$ and $x\in X.$ The fuzzy
automaton ${\cal A}^{\rm r}_\alpha$ is called the {\it
reduced  fuzzy automaton\/} associated with $\cal A$ and
$\alpha.$

\begin{theorem}\label{th:3}
Let $\cal L$ be an integral $\ell-$monoid, let $\alpha$ be an
arbitrary  fuzzy regular expression, let ${\cal A}_\alpha$ be an
arbitrary nondeterministic automaton recognizing the language $\|\alpha_R\|$, and let
${\cal A}^{\rm r}_\alpha$ be the reduced fuzzy automaton defined as in
{\rm (\ref{eq:reduced})}.~Then
\begin{equation}
L({\cal A}^{\rm r}_\alpha)=\|\alpha\|.
\end{equation}
\end{theorem}

\begin{proof}
First, we have that
\[
L({\cal A}^{\rm r}_\alpha)(\varepsilon)=\tau^{A^{\rm
r}_\alpha}(a_0)=\tau^{A_\alpha}(a_0)=\|\alpha\|(\varepsilon).
\]
Next, for every $u\in X^+$, where $u=x_1x_2\cdots x_n$, with $x_1,x_2,\dots,x_n\in X$, by (\ref{rellanguageinit}), Theorems
\ref{th:1} and \ref{th:2}, and idempotency of $R_{\cal A}$
we obtain that
\[
\begin{aligned}
\|\alpha\|(u)&=L({\cal A}_\alpha)(u)=(\delta^{A_\alpha}_{x_1}\circ\cdots\circ\delta^{A_\alpha}_{x_n}\circ\tau^{A_\alpha})(a_0)=(R_{\cal A}\circ\delta^{A}_{x_1}\circ R_{\cal
A}^2\circ\cdots\circ R_{\cal A}^2\circ\delta^{A}_{x_n}\circ
R_{\cal A}^2\circ\tau^{A})(a_0)\\
&=(R_{\cal A}\circ\delta^{A}_{x_1}\circ R_{\cal A}\circ\cdots\circ
R_{\cal A}\circ\delta^{A}_{x_n}\circ
R_{\cal A}\circ\tau^{A})(a_0)\\
&=\bigvee_{a_1,\dots,a_n\in A_\alpha}(R_{\cal
A}\circ\delta^{A}_{x_1})(a_0,a_1)\otimes\cdots\otimes(R_{\cal
A}\circ\delta^{A}_{x_n})(a_{n-1},a_n)\otimes(R_{\cal
A}\circ\tau^{A})(a_n)\\
&=^*(\delta^{A^{\rm r}_\alpha}_{x_1}\circ\cdots\circ\delta^{A^{\rm r}_\alpha}_{x_n}\circ\tau^{A^{\rm r}_\alpha})(a_0)\\
&=L({\cal A}^{\rm r}_\alpha).
\end{aligned}
\]The equality marked by $*$ follows from the fact that
$(R_{\cal A}\circ\delta^{A}_x)(a,b)=0$, for all $b\in
A_\alpha\setminus A_\alpha^{\rm r},$ and  $x\in X.$
\end{proof}

Obviously, Theorem \ref{th:3} describes a method of construction
a fuzzy automaton from a given  fuzzy~regular expression, which can be significantly smaller than the one made
by the basic construction.~Furthermore,~if the starting nondeterministic
automaton recognizing $\|\alpha_R\|$ is the position automaton
${\cal A}_{\rm p}(\alpha_R)$, then ${\cal A}^{\rm r}_{\rm
pf}(\alpha)$~has exactly $|\alpha|_X+1$ states (Example \ref{pr8}
illustrates this fact).~Accordingly, since the position automaton of a
given regular expression has the number of states equal to the
length of the considered regular expression, the fuzzy automaton
${\cal A}^{\rm r}_{\rm pf}(\alpha)$ is called the {\it position
 fuzzy automaton\/} of the given  fuzzy regular
expression $\alpha$.

\begin{example}\rm\label{pr8}
Consider a fuzzy regular expression $\alpha=(0.1x^*)(yx+0.8y)^*$
from Example \ref{pr5}.~Fuzzy transition relations
$\delta^{A_{\rm pf}}_x,$ $\delta^{A_{\rm pf}}_y$ and the  fuzzy
set $\tau^{A_{\rm pf}}$ of terminal states of the  fuzzy automaton ${\cal A}_{\rm pf}(\alpha)$ are:
\[
\delta^{A_{\rm pf}}_x=\begin{bmatrix}
                   0 & 0 & 0.1 & 0 & 0 & 0.08 & 0\\
                   0 & 0 & 1 & 0 & 0 & 0.8 & 0\\
                   0 & 0 & 1 & 0 & 0 & 0.8 & 0\\
                   0 & 0 & 0 & 0 & 1 & 0.8 & 0\\
                   0 & 0 & 0 & 0 & 0 & 0 & 0\\
                   0 & 0 & 0 & 0 & 0 & 0 & 0\\
                   0 & 0 & 0 & 0 & 0 & 0 & 0
                   \end{bmatrix},\quad
\delta^{A_{\rm pf}}_y=\begin{bmatrix}
                   0 & 0 & 0 & 0.1 & 0 & 0.064 & 0.08\\
                   0 & 0 & 0 & 1 & 0 & 0.64 & 0.8\\
                   0 & 0 & 0 & 1 & 0 & 0.64 & 0.8\\
                   0 & 0 & 0 & 0 & 0 & 0 & 0\\
                   0 & 0 & 0 & 1 & 0 & 0.64 & 0.8\\
                   0 & 0 & 0 & 0 & 0 & 0.8 & 1\\
                   0 & 0 & 0 & 1 & 0 & 0.64 & 0.8
                   \end{bmatrix},\quad
\tau^{A_{\rm pf}}=\begin{bmatrix}
                   0.1 \\
                   1 \\
                   1 \\
                   0 \\
                   1 \\
                   0 \\
                   1
                 \end{bmatrix}.
\]Evidently, $A^{\rm r}_{\rm pf}=\{0,2,3,4,6\}$, and hence, the  fuzzy finite
automaton ${\cal A}^{\rm r}_{\rm pf}$ has two states less than
the position  fuzzy automaton ${\cal A}_{\rm pf}(\alpha)$.

Fuzzy
transition relations $\delta^{A^{\rm r}_{\rm pf}}_x,$
$\delta^{A^{\rm r}_{\rm pf}}_y$, and the  fuzzy set $\tau^{A^{\rm
r}_{\rm pf}}$ of terminal states~of~${\cal A}^{\rm r}_{\rm
pf}(\alpha)$ are:
\[
\delta^{A^{\rm r}_{\rm pf}}_x=\begin{bmatrix}
                   0 & 0.1 & 0 & 0 & 0\\
                   0 & 1 & 0 & 0 & 0\\
                   0 & 0 & 0 & 1 & 0\\
                   0 & 0 & 0 & 0 & 0\\
                   0 & 0 & 0 & 0 & 0
                   \end{bmatrix},\qquad
\delta^{A^{\rm r}_{\rm pf}}_y=\begin{bmatrix}
                   0 & 0 & 0.1 & 0 & 0.08\\
                   0 & 0 & 1 & 0 & 0.8\\
                   0 & 0 & 0 & 0 & 0\\
                   0 & 0 & 1 & 0 & 0.8\\
                   0 & 0 & 1 & 0 & 0.8
                   \end{bmatrix},\qquad
\tau^{A^{\rm r}_{\rm pf}}=\begin{bmatrix}
                   0.1 \\
                   1 \\
                   0 \\
                   1 \\
                   1
                 \end{bmatrix}.
\]
\end{example}

\section{Reducing the size of position fuzzy automata by right invariant crisp equivalences}\label{sec:redstate}

The reduction of the number of states of fuzzy automata with membership values in complete
residuated lattices has been recently investigated in \cite{CSIP.07,CSIP.10,SCI.11}, where the state reduction problem has been related~to the problem of solving particular systems of fuzzy relation equations and inequalities.~Central place in~the state reduction is held by right and left invariant fuzzy equivalences, as well as by right and left invariant fuzzy quasi-orders.

Complete residuated lattices have a rich algebraic structure that provides powerful tools for solving fuzzy relational equations and inequalities, including those that define right and left invariant fuzzy~equivalences.~Unfortunately, when we deal with fuzzy automata over lattice ordered monoids we do not have~such tools, and we are forced to work with crisp equivalences.~Therefore, here we reduce the number of states of fuzzy automata over $\ell $-monoids using right invariant crisp equivalences.

The reduction of fuzzy automata over complete residuated lattices
by means of right invariant fuzzy equivalences and fuzzy
quasi-orders has been recently studied in
\cite{CSIP.07,CSIP.10,SCI.11}. It has been proved that better
state reductions can be achieved employing  right invariant fuzzy
equivalences.~Residuated lattices are~rich algebraic structures
supplied with operations called residuum and biresiduum, and
satisfying many other important algebraic properties. In some sources
residuated lattices are called integral, commutative, residuated
$\ell-$monoids. There, the operations of residuum and biresiduum
play a very important role, and are used for modelling  right
invariant fuzzy equivalences and fuzzy quasi-orders. In this
paper, however, we deal with $\ell-$monoids, in which, due to the
lack of algebraic properties and operations, construction of fuzzy
equivalence relations is a problem. Therefore, here we investigate
the problem of the reduction of  fuzzy automata by  right
invariant crisp equivalences only.~Note that the state reduction of
fuzzy automata by means of crisp equivalences has been already studied in
\cite{BG.02,CM.04,MMS.99,MM.02,P.06}, but in very special cases, and the
algorithms provided there are based on computing and merging indistinguishable
states.

Let ${\cal A}=(A,X,\delta^A,\sigma^A,\tau^A)$ be a fuzzy automaton over an $\ell $-monoid $\cal L$, and let $E$ be a fuzzy equivalence on its set of states $A$.~If $E$ is a solution to the system
\begin{equation}
\begin{aligned}\label{eq:RIFE}
&E\circ \delta^A_x\le \delta^A_x\circ E, \qquad x\in X,\\
&E\circ \tau^A=\tau^A ,
\end{aligned}
\end{equation}
then it is called a {\it right invariant\/}  fuzzy
equivalence on $\cal A$.~Dually we define {\it left invariant\/} fuzzy equivalences. A crisp equivalence on $A$ which is a
solution to (\ref{eq:RIFE}) is called a {\it right invariant crisp
equivalence\/} on $\cal A$. Note that ordinary crisp equivalences
on $A$ are considered here as  fuzzy equivalences on $A$ taking
membership values in the set $\{0,e\}\subseteq L.$

It has been shown in \cite{CSIP.10} that right invariant fuzzy
equivalences are immediate generaliza\-tions of right invariant
equivalences on nondeterministic automata, studied in a series of papers
by Ilie, Yu \cite{IY.02a,IY.02,IY.03a,IY.03,INY.04,ISY.05}, as well as in
\cite{CSY.05,CC.03,CC.04},~or
well-behaved equivalences, studied by Calude et al.~\cite{CCK.00}.
It has been also proved in \cite{CSIP.10} that congruences on
fuzzy automata, studied by Petkovi\'c in \cite{P.06}, are just
right invariant crisp equivalences on  fuzzy automata, in the
terminology from this paper.~Note that right invariant fuzzy equivalences have been called
in \cite{CIDB.11,CIJD.11} {\it forward bisimulation\/} fuzzy equivalences, whereas left invariant ones were called~{\it backward bisimulation\/} fuzzy equivalences.

In the same way as in \cite{CSIP.10} we can show that the inequality $E\circ \delta^A_x\le \delta^A_x\circ E$ is equivalent to the equation $E\circ \delta^A_x\circ E= \delta^A_x\circ E$, for each $x\in X$. We can also prove that if $E\circ \delta^A_x\le \delta^A_x\circ E$ or $E\circ \delta^A_x\circ E= \delta^A_x\circ E$ holds for every letter $x\in X$, then it also holds if we replace~the letter $x$ by an arbitrary word $u\in X^*$.

In the sequel we provide an algorithm for computing the greatest right
invariant crisp equivalence~on a fuzzy automaton with membership values in an integral $\ell-$monoid.~Such algorithm has been first given in \cite{P.06}, for~fuzzy automata over the G\"odel structure, and later in \cite{CSIP.07,CSIP.10},
for fuzzy automata over a complete residuated lattice.~The proof of the next
theorem is the same as the proof of the corresponding theorem for fuzzy automata over a complete residuated lattice, so it will be omitted.

\begin{theorem} {\bf \cite{CSIP.10,P.06}}\label{th:6}
Let $\cal L$ be an integral $\ell-$monoid, let ${\cal
A}=(A,X,\delta^A,\sigma^A,\tau^A)$ be a fuzzy finite automaton over $\cal L$. Define inductively a sequence $\{E_k\}_{k\in\Bbb N}$ of crisp
equivalences on $A$ as follows:
\begin{align}
&E_1(a,b)=\begin{cases}\ 1 & \text{if}\ \tau^A(a)=\tau^A(b)\\ \ 0 & \text{otherwise}
\end{cases},\qquad \text{for all}\ a,b\in A, \\
&E_{k+1}=E_k\land E_k^r, \qquad \text{for each}\ k\in \Bbb
N ,
\end{align}
where $E_k^r $ is a crisp equivalence on $A$ defined by
\begin{equation}
E_k^r(a,b)=\begin{cases}
\ 1\ &\ \text{if}\ \ (\delta_x\circ E_k)(a,c) = (\delta_x\circ E_k)(b,c),\ \ \text{for all}\ x\in X\ \text{and}\ c\in A \\
\ 0 &\ \text{otherwise}
\end{cases}, \qquad \text{for all}\ a,b\in A,
\end{equation}
Then the sequence $\{E_k\}_{k\in\Bbb N}$ is finite and descending, there is the least $k\in \mathbb N$ such that $E_k=E_{k+m}$, for each $m\in \mathbb
N$, and $E_k$ is the greatest right invariant crisp equivalence on the fuzzy
automaton $\cal A$.
\end{theorem}

Now we prove the following.

\begin{proposition}\label{lemanova2}
Let $\cal L$ be an integral $\ell-$monoid and let $\alpha$ be an
arbitrary  fuzzy regular expression.~For an arbitrary~non\-deterministic
automaton ${\cal A}=(A, X\cup Y,
\delta^A, a_0,\tau^A)$ recognizing the language  $\|\alpha_R\|$, the fuzzy
automaton ${\cal A}_\alpha$  associated with $\cal A$ and $\alpha$, and an arbitrary right invariant crisp equivalence $E$ on
$\cal A$ we have that
\begin{itemize}\parskip=0pt
\item[{\rm (a)}] $E\circ R_{\cal A}\le R_{\cal A}\circ E$;
\item[{\rm (b)}] $R_{{\cal A}/E}(E_a,E_b)=(R_{\cal A}\circ
E)(a,b)$\ \ for all\ \ $a,b\in A$.
\end{itemize}
\end{proposition}

\begin{proof}
(a) First, by Theorem \ref{lemanova1} and Lemmas \ref{lema1} and \ref{lema2}
we
have the following
\[
\begin{aligned}
(E\circ R_{\cal A})(a,b)&=\bigvee_{c\in A}\bigvee_{u\in
Y^*}E(a,c)\otimes\varphi^*_\alpha(u)\otimes\delta^A(c,u,b)=\bigvee_{u\in
Y^*}\varphi^*_\alpha(u)\otimes(E\circ\delta^A_u)(a,b)\\
&\le \bigvee_{u\in Y^*}\varphi^*_\alpha(u)\otimes(\delta^A_u\circ
E)(a,b)=\bigvee_{c\in A}\bigvee_{u\in
Y^*}\varphi^*_\alpha(u)\otimes\delta^A(a,u,c)\otimes
E(c,b)\\
&=(R_{\cal A}\circ E)(a,b),
\end{aligned}
\]for every $a,b\in A.$ Consequently,  $E\circ R_{\cal A}\le R_{\cal A}\circ E.$

(b) By Theorem \ref{lemanova1}, Lemmas \ref{lema1} and \ref{lema2},
and (\ref{eq:dE1}), we obtain
\[\begin{aligned}
R_{{\cal A}/E}(E_a,E_b)&=\bigvee_{u\in
Y^*}\varphi^*_\alpha(u)\otimes\delta^{A/E}(E_a,u,E_b)=\bigvee_{u\in
Y^*}\varphi^*_\alpha(u)\otimes(E\circ\delta_u^{A}\circ E)(a,b)\\
&=\bigvee_{u\in Y^*}\varphi^*_\alpha(u)\otimes(\delta_u^{A}\circ
E)(a,b)=\bigvee_{c\in A}\bigvee_{u\in
Y^*}\varphi^*_\alpha(u)\otimes\delta_u^{A}(a,c)\otimes E(c,b)\\
&=(R_{\cal A}\circ E)(a,b),
\end{aligned}\]for arbitrary $a,b\in A$.
\end{proof}

\begin{theorem}\label{th:7}
Let $\cal L$ be an integral $\ell-$monoid and let $\alpha$ be an
arbitrary  fuzzy regular expression.~Moreover, consider an arbitrary nondeterministic
automaton ${\cal A}=(A, X\cup Y,
\delta^A, a_0,\tau^A)$ recognizing the language $\|\alpha_R\|$, and the fuzzy
automaton ${\cal A}_\alpha$  associated with $\cal A$ and $\alpha$.

Then every right invariant equivalence $E$ on $\cal A$ is
a right invariant crisp equivalence on ${\cal A}_\alpha$, and the
fuzzy~auto\-maton $({\cal A}/E)_\alpha$ is isomorphic to the
factor  fuzzy automaton ${\cal A}_\alpha/E$.
\end{theorem}
\begin{proof}
Let $E$ be an arbitrary right invariant  equivalence on $\cal A$.~First, by Theorem
\ref{th:2} and statement (a) of Proposition \ref{lemanova2}
we obtain
\[
E\circ\delta^{A_\alpha}_x=E\circ R_{\cal
A}\circ\delta^{A}_x\circ R_{\cal
A}\le R_{\cal A}\circ E\circ\delta^{A}_x\circ R_{\cal A} \le R_{\cal A}\circ\delta^{A}_x\circ E\circ R_{\cal
A}\le R_{\cal A}\circ \delta^{A}_x\circ R_{\cal A}\circ
E=\delta^{A_\alpha}_x\circ E,
\]
for every $x\in X$.~In a similar way, we show that $E\circ\tau^{A_\alpha}=\tau^{A_\alpha}$.~Consequently,  $E$ is a right invariant crisp~equi\-valence on ${\cal A}_\alpha$.

Next, by (\ref{eq:dE1}), Theorem \ref{th:2}, and statement (b) of Proposition
\ref{lemanova2} we have
\[
\begin{aligned}
\delta_x^{(A/E)_\alpha}(E_a,E_b)&=(R_{{\cal
A}/E}\circ\delta_x^{A/E}\circ R_{{\cal
A}/E})(E_a,E_b)=(R_{\cal A}\circ E\circ \delta_x^A\circ E\circ R_{\cal A}\circ E)(a,b)\\&=(R_{\cal A}\circ \delta_x^A\circ R_{\cal A}\circ E)(a,b)=(\delta_x^{A_\alpha}\circ
E)(a,b)=(E\circ\delta_x^{A_\alpha}\circ
E)(a,b)=\delta_x^{A_\alpha/E}(E_a,E_b),
\end{aligned}
\]for all $x\in X$ and  $a,b\in A$.~In addition, it is easy to check that
\[
\tau^{(A/E)_\alpha}(E_a)=\tau^{A_\alpha/E}(E_a)\\
\]for every $a\in A$.
Thus the identity function on ${A/E}$ is an isomorphism from
$({\cal A}/E)_\alpha$ to ${\cal A}_\alpha/E$.
\end{proof}

According to the previous theorem,
for an arbitrary
fuzzy regular expression $\alpha$, and any nondeterministic
automaton $\cal A$ recognizing $\|\alpha\|$, the greatest right invariant
equivalence on the nondeterministic automaton $\cal A$ is less or equal to the greatest right invariant crisp equivalence on the fuzzy automaton ${\cal A}_\alpha$.~The following example shows that  even if the starting automaton $\cal A$ is a minimal deterministic automaton of the language $\|\alpha_R\|$, the fuzzy automaton ${\cal
A}_\alpha$ may be further reduced by right invariant crisp
equivalences, i.e., the greatest right invariant crisp equivalence
on ${\cal A}_\alpha$~differs from the equality relation on
$A_\alpha.$

\begin{example}\rm\label{pr10}
Let $\cal L$ be G\" odel structure, and $\alpha=x+0.5x,$ a fuzzy
regular expression over the alphabet $\{x\}$. We have that
$\alpha_R=x+\lambda x$, and the graph of the minimal
deterministic automaton $\cal A$ recognizing the language  $\|\alpha_R\|$ is presented by
Figure 3a.\begin{center}
\psset{unit=1.2cm}
\newpsobject{showgrid}{psgrid}{subgriddiv=1,griddots=10,gridlabels=6pt}
\begin{pspicture}(-0.2,-2)(11,1)

\pnode(0.5,0){I} \SpecialCoor
\rput([angle=0,nodesep=8mm,offset=0pt]I ){\cnode{3mm}{A0}}
\rput([angle=-40,nodesep=19mm,offset=0pt]A0){\cnode{3mm}{A1}}
\rput([angle=0,nodesep=30mm,offset=0pt]A0){\cnode{3mm}{A2}}
\rput([angle=0,nodesep=30mm,offset=0pt]A0){\cnode{2.5mm}{A2}}

\rput(A0){\footnotesize $0$} \rput(A1){\footnotesize $1$}
\rput(A2){\footnotesize $2$} \NormalCoor \ncline{->}{I}{A0}
\ncline{->}{A0}{A1}\Bput[1pt]{\footnotesize $\lambda$}
\ncline{->}{A1}{A2}\Bput[1pt]{\footnotesize $x$}
\ncline{->}{A0}{A2}\Aput[1pt]{\footnotesize $x$}

\pnode(6,0){J} \SpecialCoor
\rput([angle=0,nodesep=8mm,offset=0pt]J ){\cnode{3mm}{A3}}
\rput([angle=-40,nodesep=19mm,offset=0pt]A3){\cnode{3mm}{A4}}
\rput([angle=0,nodesep=30mm,offset=0pt]A3){\cnode{3mm}{A5}}

\rput(A3){\footnotesize $0$} \rput(A4){\footnotesize $1$}
\rput(A5){\footnotesize $2$}
\ncline{->}{A3}{A5}\Aput[1pt]{\footnotesize $x/1$}
\ncline{->}{A4}{A5}\Bput[1pt]{\footnotesize $x/1$}
\NormalCoor\ncline{->}{J}{A3}

\end{pspicture}\\
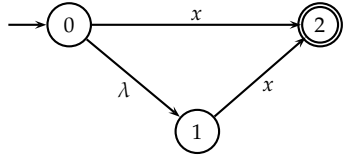
{Figure 3a. The automaton ${\cal A}$\hspace{2cm}Figure 3b. The fuzzy automaton ${\cal A}_\alpha$}
\end{center}
\medskip

Figure 3b presents the fuzzy automaton $\cal A_\alpha$
associated with $\cal A$ and $\alpha$.~The fuzzy set $\tau^{A_{\alpha}}$
of terminal~states of $\cal A_\alpha$, and the greatest right
invariant crisp equivalence $E^{\rm cri}$ on $\cal A_\alpha$ are represented
by:
\[
\tau^{A_\alpha}=\begin{bmatrix}
                   0 \\
                   0 \\
                   1
                \end{bmatrix},\quad
E^{\rm cri}=\begin{bmatrix}
                   1 & 1 & 0\\
                   1 & 1 & 0\\
                   0 & 0 & 1
            \end{bmatrix}.
\]

\end{example}

\begin{theorem}\label{th:nova1}
Let $\cal L$ be an integral $\ell-$monoid and let $\alpha$ be an
arbitrary  fuzzy regular expression.~Moreover, consider an arbitrary nondeterministic
automaton ${\cal A}=(A, X\cup Y,
\delta^A, a_0,\tau^A)$ recognizing the language $\|\alpha_R\|$, and the fuzzy
automaton ${\cal A}_\alpha$  associated with $\cal A$ and $\alpha$.

Then for an arbitrary right invariant equivalence $E$ on
$\cal A$ there exists a right invariant crisp equivalence $E^{\rm
r}$ on~${\cal A}_\alpha^r$, such that the factor  fuzzy automata
$({\cal A}/E)^{\rm r}_\alpha$ and  $({\cal A}^{\rm
r}_\alpha)/E^{\rm r}$ are isomorphic.
\end{theorem}
\begin{proof}
Let $E$ be any right invariant equivalence on
$\cal A$.~Define a crisp relation $E^{\rm r}$ on $A^{\rm
r}_{\alpha}$ by $E^{\rm r}(a,b)=E(a,b)$,
for all $a,b\in A^{\rm r}_{\alpha}$. Obviously, $E^{\rm
r}$ is a crisp equivalence  on the set $A^{\rm r}_{\alpha}$. Moreover, by (\ref{eq:reduced}) and Proposition \ref{lemanova2}, and
we have that
\[
\begin{aligned}
(E^{\rm r}\circ \delta_x^{A^{\rm r}_{\alpha}})(a,b)&=\bigvee_{c\in
A^{\rm r}_{\alpha}}E(a,c)\otimes(R_{{\cal
A}}\circ\delta_x^{A})(c,b)\le\bigvee_{c\in A_{\alpha}}E(a,c)\otimes(R_{{\cal A}}\circ\delta_x^{A})(c,b)\\
&=(E\circ R_{\cal A}\circ\delta_x^{A})(a,b)\le (R_{\cal
A}\circ\delta_x^{A}\circ E)(a,b)=^*(\delta_x^{A^{\rm
r}_{\alpha}}\circ E^{\rm r})(a,b),
\end{aligned}
\]for all $a,b\in A^{\rm r}_{\alpha}.$ Note that the equality marked by $*$ follows from the fact
that $(R_{{\cal A}_\alpha}\circ\delta^{A_\alpha}_x)(a,b)=0$, for
every $b\in A_\alpha\setminus A_{\alpha}^{\rm r},$ and every $x\in
X.$ Thus $E^{\rm r}$ is a right invariant crisp equivalence on
${\cal A}_\alpha^{\rm r}.$

Define a mapping $\Phi:A^{\rm r}_\alpha/E^{\rm r}\to (A/E)^{\rm
r}_\alpha$ by $\Phi(E^{\rm r}_a)=E_a$, for every $E^{\rm r}_a\in A^{\rm r}_\alpha/E^{\rm
r}$.~For an arbitrary $E^{\rm r}_a\in A^{\rm r}_\alpha/E^{\rm r}$,
there are $b\in A$ and $x\in X$ such that $\delta^A(b,x,a)=1$.~Since~we have
that $\delta^{A/E}(E_b,x,E_a)\ge \delta^A(b,x,a)$, we
obtain $E_a\in (A/E)^{\rm r}_\alpha.$ Moreover, from
\[
E^{\rm r}_a=E^{\rm r}_b\iff E^{\rm r}(a,b)=1\iff E(a,b)=1\iff
E_a=E_b,
\]for all $a,b\in A^{\rm r}_\alpha,$ we conclude that $\Phi$ is
both a well-defined and an injective mapping.

 Further, for any $E^a\in
(A/E)^{\rm r}_\alpha$ there are $E_b\in A/E$ and $x\in X$ such
that $\delta^{A/E}(E_b,x,E_a)=1$, which implies
\[
\delta^A(c,x,d)=1,\ E_b=E_c,\ E_d=E_a,\text{\ \ for some\ \ }c,d\in
A.
\]
Thus $d\in A^{\rm r}_\alpha$, and $\Phi(E^{\rm r}_d)=E_d=E_a.$ In
conclusion, $\Phi$ is a bijective mapping.

Finally, by (\ref{eq:dE1}), (\ref{eq:reduced}), and Proposition
\ref{lemanova2}, we have
\[
\begin{aligned}
\delta_x^{A^{\rm r}_\alpha/E^{\rm r}}(E^{\rm r}_a,E^{\rm
r}_b)&=(\delta_x^{A^{\rm r}_\alpha}\circ E^{\rm
r})(a,b)=\bigvee_{c\in A^{\rm r}_\alpha}(R_{\cal
A}\circ\delta^A_x)(a,c)\otimes E^{\rm r}(c,b)\\
&=\bigvee_{c\in A}(R_{\cal A}\circ\delta^A_x)(a,c)\otimes E^{\rm
r}(c,b)=(R_{\cal A}\circ\delta^A_x\circ E)(a,b)\\
&=(R_{\cal A}\circ E\circ\delta^A_x\circ E)(a,b)=(R_{{\cal
A}/E}\circ\delta^{A/E}_x)(E_a,E_b)=\delta_x^{(A/E)^{\rm r}_\alpha}(\Phi(E^{\rm r}_a),\Phi(E^{\rm
r}_b))
\end{aligned}
\]for all $E^{\rm r}_a,E^{\rm
r}_b\in {\cal A}^{\rm r}_\alpha/E$. Therefore, $\Phi $ is an isomorphism.
\end{proof}

Let us recall that Theorem \ref{th:3} gives us a simple method to
construct various types of  fuzzy automata from the fuzzy regular expression
$\alpha$.~This method is based on choice
 of different
nondeterministic automata~$\cal A$ recognizing $\|\alpha_R\|$,
from which we obtain different  fuzzy automata ${\cal
A}_\alpha$ recognizing $\|\alpha\|$.

Let $\cal L$ be an integral $\ell-$monoid, and let $\alpha$ be an
arbitrary  fuzzy regular expression over an alphabet $X$. For a
regular expression $\alpha_R$ over $X\cup Y$, where $Y$ is an
alphabet associated with $\alpha,$ let  ${\cal A}_{\rm
f}(\alpha_R)=(A_{\rm f}, X\cup Y, \delta^{A_{\rm f}}, 0,
\tau^{A_{\rm f}})$ be the follow automaton of $\alpha_R$. In this
paper we will assume that the the follow automaton of $\alpha$ is
exactly the factor automaton of the position automaton of $\alpha$
with respect to a particular right invariant  equivalence~$E$, called
the {\it follow equivalence\/}. For the definition of the follow
equivalence we refer to \cite{IY.02a,IY.02,IY.03a,IY.03}.
Starting from ${\cal A}_{\rm f}(\alpha_R)$, by (\ref{eq:fuzzyrec.1}), (\ref{eq:fuzzyrec.2})
and (\ref{eq:reduced}) we obtain the reduced  fuzzy automaton associated with ${\cal A}_{\rm f}(\alpha_R)$ and
$\alpha$, which is denoted by ${\cal A}^{\rm r}_{\rm
ff}(\alpha)=(A_{\rm ff}, X, \delta^{A_{\rm ff}}, 0, \tau^{A_{\rm
ff}})$. The fuzzy automaton ${\cal A}^{\rm r}_{\rm
ff}(\alpha)$ is called the {\it follow  fuzzy automaton\/}~of~$\alpha$.

\begin{theorem}\label{th:8}
Let $\cal L$ be an integral $\ell-$monoid, let $\alpha$ be an
arbitrary  fuzzy regular expression, and let ${\cal A}^{\rm r}_{\rm
pf}(\alpha)$ and ${\cal A}^{\rm r}_{\rm ff}(\alpha)$ be respectively the position  fuzzy  automaton and the follow fuzzy automaton of $\alpha$.

Then the follow  fuzzy automaton ${\cal A}^{\rm r}_{\rm ff}(\alpha)$ of $\alpha$ is isomorphic to the
factor  fuzzy automaton of the position  fuzzy automaton
${\cal A}^{\rm r}_{\rm pf}(\alpha)$ of $\alpha$ with respect to some right invariant crisp equivalence on ${\cal A}^{\rm r}_{\rm pf}$.
\end{theorem}
\begin{proof}
The proof is an immediate consequence of Theorem \ref{th:nova1}.
\end{proof}
By Theorem \ref{th:8} we obtain that the follow  fuzzy automaton is
the reduced position  fuzzy automaton with respect to some right invariant crisp
equivalence, and therefore, it may be significantly smaller.~However,
Example \ref{pr10} shows that, in the general case, follow equivalences are
not necessarily the greatest right invariant crisp equivalences on the
position  fuzzy automata. Consequently, smaller  fuzzy automata
from a given  $\alpha$ can be obtained by reducing the size of
the position  fuzzy automaton of $\alpha$ by means of the greatest right
invariant crisp equivalence.

\begin{example}\rm\label{pr11}
Let $\cal L$ be G\" odel structure. Consider $\alpha=xx^*+0.1x^*$,
the  fuzzy regular expression over the alphabet $X=\{x\}$. An
expression $\alpha_R=xx^*+\lambda x^*$, over the alphabet
$\{x,\lambda\}$, is the regular expression obtained from $\alpha$.

The position automaton ${\cal A}_{\rm p}(\alpha_R)$ is given by
the following  fuzzy transition relations
\[
\delta^{A_{\rm p}}_x=\begin{bmatrix}
                   0 & 1 & 0 & 0 & 0\\
                   0 & 0 & 1 & 0 & 0\\
                   0 & 0 & 1 & 0 & 0\\
                   0 & 0 & 0 & 0 & 1\\
                   0 & 0 & 0 & 0 & 1
                   \end{bmatrix},\qquad
\delta^{A_{\rm p}}_\lambda=\begin{bmatrix}
                   0 & 0 & 0 & 1 & 0\\
                   0 & 0 & 0 & 0 & 0\\
                   0 & 0 & 0 & 0 & 0\\
                   0 & 0 & 0 & 0 & 0\\
                   0 & 0 & 0 & 0 & 0
                   \end{bmatrix},\qquad
\tau^{A_{\rm p}}=\begin{bmatrix}
                   0 \\
                   1 \\
                   1 \\
                   1 \\
                   1
                 \end{bmatrix},
\]
and the position  fuzzy automaton ${\cal A}^{\rm r}_{\rm
pf}(\alpha)$ is given by the following  fuzzy transition relations
\[
\delta^{A^{\rm r}_{\rm pf}}_x=\begin{bmatrix}
                   0 & 1 & 0 & 0.1\\
                   0 & 0 & 1 & 0\\
                   0 & 0 & 1 & 0\\
                   0 & 0 & 0 & 1
                   \end{bmatrix},\text{\ \ and\ \ }
\tau^{A^{\rm r}_{\rm pf}}=\begin{bmatrix}
                   1 \\
                   1 \\
                   1 \\
                   1
                 \end{bmatrix}.
\]
The follow relation $E_{\rm f}$ on ${\cal A}_{\rm p}(\alpha_R)$,
and the related right invariant crisp  equivalence $E_{\rm f}^{\rm
r}$ on $A^{\rm r}_{\rm pf}(\alpha)$ are
\[
E_{\rm f}=\begin{bmatrix}
                   1 & 0 & 0 & 0 & 0\\
                   0 & 1 & 1 & 0 & 0\\
                   0 & 1 & 1 & 0 & 0\\
                   0 & 0 & 0 & 1 & 1\\
                   0 & 0 & 0 & 1 & 1
                   \end{bmatrix},
E^{\rm r}_{\rm f}=\begin{bmatrix}
                   1 & 0 & 0 & 0\\
                   0 & 1 & 1 & 0\\
                   0 & 1 & 1 & 0\\
                   0 & 0 & 0 & 1
                   \end{bmatrix},
\]and therefore follow  fuzzy automaton ${\cal A}^{\rm r}_{\rm f{\rm f}}(\alpha)$ has 3
states. However, since the greatest right invariant crisp
equivalence $E_1^{\rm{cri}}$ on ${\cal A}^{\rm r}_{\rm
pf}(\alpha)$ is given by
\[
E_1^{\rm{cri}}=\begin{bmatrix}
                   1 & 1 & 1 & 1\\
                   1 & 1 & 1 & 1\\
                   1 & 1 & 1 & 1\\
                   1 & 1 & 1 & 1
                   \end{bmatrix}.
\]
we conclude that the  fuzzy finite automaton ${\cal A}^{\rm
r}_{\rm pf}(\alpha)/E_1^{\rm cri}$ has only 1 state, and is
significantly smaller than ${\cal A}^{\rm r}_{\rm pf}(\alpha)$.
\end{example}

Let $\alpha$ be an   regular expression.~Observe that, starting
from the partial derivative automaton of the~regular expression $\alpha_R$
obtained
from $\alpha$, it is possible to construct the fuzzy partial
derivative automaton of $\alpha$. Since the partial derivative
automaton is isomorphic to the factor automaton of the position
automaton with respect to certain right invariant equivalence (cf.~\cite{CZ.01a,CZ.01b,CZ.02,IY.03a}), the result which
correspond to Theorem \ref{th:8}, concerning  fuzzy partial
derivative automata, can be easily derived.

\section{Concluding remarks}

In this paper we have discussed the problem of the effective construction of a fuzzy finite automaton from a given fuzzy regular expression.~We have approached this problem by converting a given fuzzy regular expression $\alpha $ over an alphabet $X$ in an ordinary regular expression $\alpha_R$ over a larger alphabet $X\cup Y$ obtained by adding new letters assigned to different scalars that appear in the fuzzy regular expression~$\alpha $. Starting from an arbitrary nondeterministic finite automaton $\cal A$ that recognizes the language $\|\alpha_R\|$ represented by the regular expression $\alpha_R$, we have constructed a fuzzy finite automaton ${\cal A}_\alpha $ associated with~$\cal A$ and $\alpha $, which recognizes the fuzzy language $\|\alpha\|$ represented by $\alpha $.~The starting nondeterministic finite~auto\-maton $\cal A$ can be
obtained from $\alpha_R$ using any of the well-known constructions for converting regular expressions to nondeterministic finite automata, such as Glushkov-McNaughton-Yamada's position automaton, Brzozowski's derivative automaton, Antimirov's partial derivative automaton, or Ilie-Yu's follow automaton.

The fuzzy finite automaton ${\cal A}_\alpha $ that we have constructed has the same number of states as the starting nondeterministic finite automaton $\cal A$, but we have also given the reduced version of the fuzzy automaton ${\cal A}_\alpha $ which can have strictly less number of states than $\cal A$.~Moreover, we have discussed the reduction of the number of states of the fuzzy automaton ${\cal A}_\alpha $ by means of right invariant crisp equivalences.

All the main results of the paper have been proved for fuzzy automata taking membership values in an integral lattice-ordered monoiod.

\end{document}